\documentclass[11pt, letterpaper]{article}
\usepackage[letterpaper, top=1in, bottom=1in, left=1in, right=1in]{geometry}
\usepackage[hyphens]{url}  
\usepackage{graphicx} 
\usepackage{natbib}  
\usepackage{caption} 
\usepackage{amsmath}
\usepackage{algorithm}
\usepackage{algorithmic}
\usepackage{booktabs}

\usepackage{amsthm,amsmath,amsfonts,amssymb}
\usepackage{dsfont}

\newtheorem{assumption}{Assumption}
\newtheorem{theorem}{Theorem}
\newtheorem{proposition}{Proposition}
\newtheorem{lemma}{Lemma}
\newtheorem{corollary}{Corollary}

\usepackage{xcolor}

\newcommand{\Prob}{\mathbb{P}}

\newcommand{\footremember}[2]{%
    \footnote{#2}
    \newcounter{#1}
    \setcounter{#1}{\value{footnote}}%
}

\title{Effort Matters in Score-Based Admissions: 
\\How Retaking and Aggregation Shape Test Scores}
\author{Christine Ling\footremember{trailer}{Georgia Institute of Technology. Email: cling39@gatech.edu. 
}
\and Diptangshu Sen\footremember{alley}{Georgia Institute of Technology. Email: dsen30@gatech.edu }
\and Juba Ziani\footremember{somethingelse2}{Georgia Institute of Technology. Email: jziani3@gatech.edu}
}

\date{\today}

\begin{document}

\maketitle

\begin{abstract}
Observed standardized test scores are the result of an endogenous process: students strategically allocate effort across multiple retake attempts to improve their outcomes. Because students differ in their ability to make these investments, the interaction between applicant strategy and institutional scoring rules---such as the widely used Single-Sitting and Superscoring policies---can disparately distort observed scores.

We develop a strategic framework where students allocate effort in response to different scoring policies. We show that Superscoring---the practice of combining the best section scores across attempts---introduces systematic score inflation through order-statistic selection over noise draws. This degrades signal accuracy and amplifies wealth-based disparities by disproportionately rewarding applicants who can afford repeated testing. Conversely, Single-Sitting---which keeps the best overall score rather than section-level scores---preserves signal fidelity but excludes high-ability students who lack the resources to prepare for all subjects simultaneously. Neither rule uniformly dominates; instead, they force a structural trade-off between statistical precision and fair outcomes. Finally, to address this, we propose three algorithmic interventions which either modify how scores from multiple attempts are combined, or apply a post-hoc correction to observed scores. Using simulations calibrated to 2025 College Board data, we compare standard scoring rules against these proposed interventions. 
\end{abstract}

\section{Introduction}
\label{sec:intro}
Standardized testing is a central component of university admissions. As institutions increasingly rely on algorithmic screening to process large applicant pools, the mathematical rules used to aggregate test scores can substantially influence admission outcomes. Two policies currently dominate. Under \emph{Single-Sitting}, institutions evaluate an applicant's best complete exam attempt. Under \emph{Superscoring}, institutions independently combine the best section scores obtained across multiple attempts. Although these policies are typically viewed as passive measurements, they fundamentally alter the incentives students face. Observed test scores are therefore the outcome of a strategic process in which students decide how and when to invest effort, rather than an exogenous measurement of ability.

We formalize standardized testing as an incentive problem. Rather than treating test scores as exogenous signals, we model students as strategic agents who allocate study effort across multiple subjects and retake opportunities in response to the institution's aggregation rule. Because additional preparation and repeated testing incur both financial and temporal costs, students' ability to optimize against the scoring algorithm depends on their available resources. Consequently, observed scores reflect the joint interaction of intrinsic ability, strategic effort, and wealth constraints.

We study how strategic student behavior and institutional scoring rules jointly shape both observed test scores and downstream admissions decisions. Because students optimize against the aggregation rule, changing how scores are combined changes not only how applicants are evaluated, but also how they prepare for the exam and whether they choose to retake it. These strategic responses fundamentally alter the statistical properties of the scores observed by the institution, making score aggregation both a measurement problem and an incentive design problem. Our goal is to understand these interactions, characterize the trade-offs induced by existing aggregation rules, and study whether alternative scoring mechanisms can better balance the competing objectives of accurate measurement of student ability, broad participation, and fair outcomes. Our main contributions are:
\begin{enumerate}
    \item \textbf{Endogenous Effort-Based Testing Model.} In Section~\ref{sec:model}, we provide a novel model of standardized test-taking as a two-stage optimization problem where students strategically allocate study effort across multiple subjects and attempts, taking wealth disparities into account.
    \item \textbf{Strategic Responses to Scoring Rules.} In Section~\ref{sec:student_behavior}, we demonstrate that scoring rules fundamentally alter how students invest effort into preparation. We show that Superscoring induces effort specialization driven by favorable test-day noise, whereas Single-Sitting imposes a simultaneous preparation burden that disparately harms the scores of wealth-constrained students.
    \item \textbf{Measurement Distortions.} In Section~\ref{sec:measurement}, we establish that Superscoring systematically inflates scores and degrades signal fidelity relative to Single-Sitting. We formally define the ``Access Premium,'' showing that Superscoring amplifies existing inequalities by rewarding affluent applicants who can afford multiple test attempts.
    \item \textbf{Empirical Evaluation of Algorithmic Corrections.} Using simulations empirically calibrated to the 2025 College Board data, we computationally evaluate both standard scoring rules and our proposed alternative interventions in Section~\ref{sec:experiments_and_corrections}. We demonstrate experimentally that while demographic-blind statistical mechanisms penalize unearned noise, achieving ability-proportionate representation requires demographic-aware structural interventions.
\end{enumerate}
Finally, we provide a discussion of our limitations and proposed future work in Appendix~\ref{app:futurework}.

\section{Related Work}
\label{sec:related}

\paragraph{Empirical Frictions in Admissions} Prior work documents correlations between parental wealth and student achievement \cite{chetty2014where}, shows that elite admissions preferences further widen these disparities \cite{chetty2025diversifying}, and finds that standardized test scores and application essays are structurally tied to household income \cite{alvero2021essay}. Financial and logistical barriers---including test fees and access to preparation resources---further discourage qualified low-income applicants \cite{klasik2012college, bulman2015effect}, contributing to the population of high-achieving ``missing one-offs'' who never enter the applicant pool \cite{hoxby2012missing, sackett2012role}. 

\paragraph{Strategic Classification and Performative Prediction}
We build on strategic classification, where agents modify observable features to improve algorithmic outcomes. This literature distinguishes between scenarios where agents ``game'' a classifier through superficial changes \cite{hardt2016strategic, braverman2020role, dong2018strategic}, and settings where agents invest effort to make genuine attribute improvements \cite{kleinberg2020how, bechavod2021learning, harris2021stateful,efthymiou2026incentivizing}. We draw a parallel in admissions, demonstrating how different scoring rules incentivize genuine learning through study effort and variance harvesting through repeated test attempts. Because universities enforce capacity constraints, our model expands strategic ranking \cite{liu2022strategic}---where agents compete for relative placement rather than absolute classification---by incorporating endogenous wealth frictions. We build on foundational work addressing fairness in strategic classification \cite{milli2019social, hu2019disparate, estornell2023fairness} and score variance \cite{garg2021standardized} to show how unequal wealth budgets for retaking exams compound systemic inequalities.

\paragraph{Information Design and Strategic Manipulation.} \cite{frankel2022improving} design optimal systems that correct for expected manipulation, matching the motivation behind our Full-Information Score Correction (FISC). In admissions contests, \cite{krishna2022pareto} show that hiding score information can improve overall effort, but assume homogeneous applicants. By incorporating wealth inequality, our setting also aligns with work showing that unequal manipulation costs fundamentally affect fairness \cite{hu2019disparate, liu2019disparate}, and that unequal access to signaling mechanisms drives downstream inequality \cite{immorlica2019access}. Our Contextual Correction relates to the trade-off between predictive accuracy and demographic fairness established by \cite{kleinberg2018human}.

\paragraph{Strategic Retaking and Test Design} \cite{vigdor2003retaking} show that taking the highest submitted score creates unequal incentives to retake the test, disproportionately benefiting students who can afford repeated attempts.\cite{zwick2019assessment} documents the resulting institutional shift toward Superscoring. \cite{rios2026admission} studies the efficiency-equity trade-offs of retake policies using administrative data, but treats study effort as fixed. Recent work has also examined the broader algorithmic design and equity of SAT-based admissions \cite{niu2022best, borghesan2022heterogeneous, george2026design}. Much of this literature focuses on test-optional policies \cite{garg2026dropping, liu2021test, dessein2024test}, while we study how standardized test scores should be aggregated when testing remains a required component of the admissions process.

\section{Model}
\label{sec:model}

We now introduce a model of standardized test-taking where students' final test scores are influenced endogenously by how students from different economic backgrounds strategically respond to the university's score aggregation policy under test-affordability constraints. The university uses a student's final score as a proxy for their true ability to make admission decisions. In our setting, there is a single admissions cycle during which students choose how to allocate effort across multiple subjects/dimensions and a maximum of two attempts to maximize their chances of admission. 

\subsection*{Student Model}
\paragraph{Population and Ability} 
We consider a population of students indexed by $i \in \{1,...,N\}$. Each student $i$ is characterized by a socioeconomic group label $g_i \in \mathcal{G}$ (observable), wealth level $W_i$ (unobservable), and a two-dimensional type $(\eta_{iM},\eta_{iV})$ (unobservable). A student's type represents their underlying ability, i.e., how efficiently effort translates into academic ability in Mathematics ($\eta_{iM}$) and Verbal (or ``Evidence-Based Reading and Writing'') ($\eta_{iV}$)\footnote{Note that we use "Math" and "Verbal" for ease of exposition.}. Wealth $W_i$ represents the total budget available for test preparation and retakes. We drop the subscript $i$ when clear from context.

\paragraph{Effort and Costs} 
Students exert effort $x_t$ (Math) and $y_t$ (Verbal) in attempt $t \in \{1, 2\}$. A student always takes the exam at $t=1$. They may opt to retake at $t=2$. Letting $e_t=x_t+y_t$ denote total effort in attempt $t$, students incur a variable cost of effort $c(e_t\mid W)$. Students who retake the exam incur an additional fixed cost $C(W)$. This cost depends on a student's wealth, modeling the fact that wealthier students may face lower barriers to entry for another attempt---such as transportation, scheduling flexibility, and the ability to take time off work---while students with fewer resources may find these fixed burdens prohibitive.
 
\begin{assumption}
\label{ass:costs}
We assume variable costs take the form $c(e_t\mid W) = k(W) \cdot e_t^\alpha$, where $k(\cdot)>0$ is strictly decreasing in wealth, and $\alpha \ge 1$ controls cost curvature.
\end{assumption}

This cost function isolates two effects: $e_t^\alpha$ enforces diminishing returns to effort, and $k(W)$ ensures wealthier students face strictly lower marginal costs to exert the same effort. Power law cost functions for effort investment have been used extensively in the political economy literature, particularly in contest theory and rent-seeking games~\citep{corchon2007theory}.

\paragraph{Score Generation} 
A student's exerted effort translates into a noisy score. Investing effort $(x_1,y_1)$ in round 1 obtains:
\begin{align*}
    X_1(x_1) = f(\eta_M x_1) + \delta_{M1}, \; Y_1(y_1) = f(\eta_V y_1)+\delta_{V1}.
\end{align*}
If the student elects to retake the exam and invests additional effort $(x_2,y_2)$ in round 2, their scores are then:
\begin{align*}
&X_2(x_1,x_2) = f(\eta_M (x_1+x_2))+\delta_{M2}, 
\\&Y_2(x_1,x_2) = f(\eta_V (y_1+y_2))+\delta_{V2},
\end{align*}
where $X$ and $Y$ denote Math and Verbal scores respectively, $f$ is a real-valued function, and $\delta_t = (\delta_{Mt}, \delta_{Vt})$ is independent noise drawn from some continuous distribution with $\mathbb{E}[\delta_t]=\textbf{0}, \text{Var}(\delta)=\sigma^2 I_2$. 
During the retake, the student retains all the benefit from effort invested in the first attempt.   

A student's score depends jointly on intrinsic ability and effort invested into preparation. Letting $\eta_j$ be the rate at which a student converts effort into academic proficiency for subject $j \in \{M,V\}$, $\eta_M x$ (respectively $\eta_V y$ for Verbal) represents academic proficiency in Math. Then the function $f$ translates this proficiency into a deterministic score. We make the following assumption on $f$:

\begin{assumption}
The function $f: \mathbb{R}_+ \to \mathbb{R}_+$ satisfies $f(0)=0$, $f'(z)>0$, and $f''(z)<0$, i.e., it is monotonically increasing and strictly concave.
\end{assumption}
The above assumption means that increased proficiency does translate into higher SAT scores, but with eventual diminishing returns on additional effort.

\subsection*{University Decisions}
We consider an institution that may put different weights on the Mathematics and Verbal categories. Let $w_M,w_V \in (0,1)$ represent the weights placed on these subjects, such that $w_M+w_V=1$. The university seeks to admit students with the highest underlying ability $(\eta_M,\eta_V)$. However, because intrinsic ability is unobservable, admissions decisions rely on submitted scores as a proxy. While the university can only observe the submitted scores, they have flexibility in the design of which scores they observe. We consider two dominant scoring rules currently used for aggregating scores:

\begin{itemize}
\item Single-Sitting (SS): Students submit one complete examination per attempt. The submitted composite score equals the better of the complete sittings:
\begin{align*}
    S^{SS} = \max \left(w_M X_1 + w_V Y_1,w_M X_2 + w_V Y_2 \right).
\end{align*}
\item Super-scoring (SC): The institution independently selects the highest \emph{section} score obtained across all attempts:
\begin{align*}
    S^{SC} = w_M \max (X_1,X_2) + w_V\max (Y_1,Y_2).
\end{align*}
\end{itemize}
We denote $S^R((x_1,y_1),(x_2,y_2))$ as the random variable representing the final score of a student under rule $R$. If a student takes the test only once (where round 2 is empty), the score under either rule trivially reduces to a single realization:
\[
    S^R((x_1,y_1), \emptyset) \triangleq w_M X_1 + w_V Y_1.
\]
The university admits applicants whose submitted scores exceed a threshold $S_{admit}^R$. Defining the composite ability of a student as $\eta\equiv w_M\eta_{M} + w_V\eta_{V}$, the institute aims to maximize the expected ability of the admitted class: 
\[
    \max_R \mathbb{E}[\eta \mid S^R \ge S_{admit}^R].
 \]

\subsection*{Student Optimization Problem}
Finally, we study the optimization problem the students face. In particular, we consider strategic students who optimize how much effort they exert as a function of their underlying type, wealth, and the university's choice of scoring rule.

A student must decide how to allocate effort across time steps $t \in \{1, 2\}$\footnote{However, effort at time step 2 can be zero.}. Each student has a quasi-linear utility defined by the final score they achieve minus the total cost of effort and, if applicable, the fixed retake cost. To model time preference, we introduce a discount factor $\beta \in (0, 1]$.
\begin{align*}
    u_1(x_1,y_1) = S^R((x_1,y_1),\emptyset) - c(x_1+y_1 \mid W).
\end{align*}

If a student takes the exam in both rounds, their utility is: 
\begin{align*}
    u_{12}((x_1,y_1),(x_2,y_2)) 
    = S^R((x_1,y_1),(x_2,y_2)) 
    - c(x_1+y_1 \mid W) 
    - \beta c(x_2+y_2 \mid W) - \beta C(W),
\end{align*}
where $C(W)$ is the retake cost\footnote{We drop the dependency in the wealth when clear from context.}$^{,}$\footnote{While the utility of the final admission score remains constant, future effort and financial costs are discounted relative to the present}.

\paragraph{Order of Operations} 
After observing the realized first-round score $S^R((x_1,y_1),\emptyset)$, a student retakes the exam if the expected score increase from a second attempt exceeds the discounted variable and fixed costs; i.e., if and only if
\begin{align*}
    \max_{x_2,y_2}~~~\mathbb{E} 
    \left[S^R((x_1,y_1),(x_2,y_2)) | x_1, y_1 \right] 
     - \beta c(x_2+y_2 \mid W) - \beta C(W) 
     \geq  S^R((x_1,y_1),\emptyset).
\end{align*}
Together, the Stage 2 optimization program is: 
\begin{align} 
\label{eq:V2}
V_2^R(x_1,y_1) = \max_{x_2,y_2 \ge 0} \Big\{ \mathbb{E}_{\delta_2} \left[ S^R((x_1,y_1),(x_2,y_2)) \right]
- \beta c(x_2+y_2 \mid W) - \beta C(W) \Big\}.
\end{align}
Anticipating this optimal second-stage behavior, the ex-ante optimization problem in round 1 is:

\begin{align}
\label{eq:V1}
    V_1^R = \max_{x_1,y_1 \ge 0} \Big\{ \mathbb{E}_{\delta_1} \Big[ \max \Big( S^R((x_1,y_1),\emptyset), V_2^R(x_1,y_1) \Big) \Big] 
    - c(x_1+y_1 \mid W) \Big\}.
\end{align}

We denote $(x_1^*,y_1^*)$ as the optimal effort in round 1 (the solution to Program~\ref{eq:V1}) and $(x_2^*,y_2^*)$ as the optimal effort in round 2 (the solution to Program~\ref{eq:V2}).

 \section{Characterization of Student Behavior}
 \label{sec:student_behavior}

We first analyze student behavior under linear effort costs ($\alpha = 1$) for analytical tractability and simplicity of exposition. Our core structural properties---including the Safety Net Effect (Proposition~\ref{prop:safetynet}(i)), Monotonicity in Retake Costs (Proposition~\ref{prop:monotone}), and Effort Specialization (Proposition~\ref{prop:specialization})---largely extend to strictly convex costs ($\alpha > 1$). Formal extensions for convex costs are provided in Appendix~\ref{app:convex_robustness}.

We make the following assumption throughout:

\begin{assumption}
\label{ass:viable}
We assume $\frac{\beta k}{w_j\eta_j}\le f'(0)$ $\forall j\in\{M,V\}$.
\end{assumption}
We show in Appendix~\ref{app:target_lemma} that this assumption is exactly equivalent to exerting non-trivial (strictly positive) effort.

\paragraph{Monotonicity in Retake Cost} 
Before characterizing how the two rules differ, we note that across both rules, effort decisions are monotone in the fixed retake cost $C$.
\begin{proposition}[Monotonicity in Fixed Retake Costs]
\label{prop:monotone}
Assume an interior Stage 1 optimum. Under either rule $R\in\{SS,SC\}$, an exogenous increase in the fixed retake cost $C$ strictly increases initial effort in both subjects.

\end{proposition}
The complete proof is in Appendix~\ref{subsec:proofmonotone}. 
A higher fixed price reduces the student’s reliance on Stage 2 as a safety net. Since $C(W_i)$ strictly decreases in wealth, wealthier populations have the capacity to explicitly budget for retakes.

\paragraph{Divergent Effort Allocation} 
The two rules diverge when a student's initial preparation falls short of the target defined in Lemma \ref{lem:target}, but they receive a favorable testing noise.

\begin{proposition}[Effort Specialization]
\label{prop:specialization}
Suppose a student elects to retake in Round 2. There exist $X^*,~Y^*$ s.t.
\begin{itemize}
\item \textit{Single-Sitting.} Second-round effort remains strictly postive for both subjects $(x_2^*>0 \text{ and } y_2^*>0)$, even if a realized section score from Round 1 exceeds its target (e.g., $X_1 \ge X^*$).
\item \textit{Superscoring.} If the realized score satisfies $X_1 \ge X^*$ for Mathematics (or $Y_1 \ge Y^*$ for Verbal), all effort concentrates in the deficient subject.
\end{itemize}
$X^*$ and $Y^*$ are called ``target scores'' for Mathematics and Verbal and are fully characterized in Appendix~\ref{app:target_lemma}.
\end{proposition}
The derivation of these corner solutions is detailed in Appendix~\ref{subsec:proofspecialization}. Under Single-Sitting, students cannot bank a high section score independently; a retake requires a new attempt, forcing continued effort across both subjects to protect against new test-day noise. Under Superscoring, the section-level max operator preserves favorable noise, generating effort specialization driven by lucky realizations, allowing the student to focus remaining effort on the deficient subject.

\paragraph{The Safety Net Effect}
Because Superscoring preserves previously earned section scores, it reduces the downside of an imperfect first sitting.
\begin{proposition}[Safety Net Effect]
\label{prop:safetynet}
Let $(x_1^{SS},y_1^{SS})$ and $(x_1^{SC},y_1^{SC})$ denote equilibrium first-round effort.
\begin{itemize}
    \item[(i)] For any realized scores $(X_1,Y_1)$, the expected continuation utility of retaking is higher under Superscoring:
\[ 
        V_2^{SC}(X_1,Y_1)\ge V_2^{SS}(X_1,Y_1).
    \]
    \item[(ii)] In equilibrium, this safety net reduces first-round effort:
\[ 
        x_1^{SC}\le x_1^{SS} \quad\text{and}\quad y_1^{SC}\le y_1^{SS}.
\]
\end{itemize}
\end{proposition}
The formal proof is provided in Appendix~\ref{subsec:proofsafetynet}. Under linear effort costs, students rationally respond to the Superscoring safety net by reducing precautionary first-round effort. The insurance that lowers the risk of testing also weakens the initial preparation incentive for the entire applicant pool.

\section{Consequences for Measurement}
\label{sec:measurement}
The previous section characterizes student responses to score aggregation rules. We now study how these responses affect the statistical properties of scores observed by universities. 

\subsection{The Impact of Superscoring on Score Inflation and Signal Fidelity}
\label{subsec:inflation}
In this section, we study how Superscoring affects score inflation and bias compared to a Single-Sitting baseline, and how these effects are disparate across students of varying wealth. The score inflation and disparity results in Section~\ref{subsec:inflation} depend primarily on order statistics, holding generally across cost structures ($\alpha \ge 1$). We invoke strict convexity ($\alpha > 1$) in Section~\ref{subsec:institutetradeoff} to analyze participation bounds.

\paragraph{Score Inflation and Signal Degradation} 
We first show how Superscoring leads to higher score distortion and inflation compared to Single-Sitting.

\begin{proposition}[Screening Distortion] 
\label{prop:screening}
Under non-degenerate noise $(\mathrm{Var}(\delta)>0)$ and $N_i\ge 2$ attempts, the expected submitted score under Superscoring strictly exceeds that of Single-Sitting, and both strictly exceed the expected score of a single attempt:
\begin{align*}
    \underbrace{\mathbb{E}\left[w_M\max_k X_k+w_V\max_k Y_k\right]}_{\text{Superscoring}}
     > \underbrace{\mathbb{E}\left[\max_k w_MX_k+w_VY_k \right]}_{\text{Single-Sitting}}
    > \underbrace{\mathbb{E}\left[w_M X_1 + w_V Y_1\right]}_{\text{Single Attempt}}.
\end{align*}
\end{proposition}
The proof is provided in Appendix~\ref{subsec:propscreening}. Both multiple-attempt policies inflate expected scores above what a student would achieve in a single attempt (which serves as the baseline measurement of objective ability). However, Superscoring inflates the score strictly more. While Single-Sitting takes the maximum over complete attempts, Superscoring takes the maximum per section. This decouples noise across subjects, allowing a composite score to exceed any single day's actual performance. Consequently, higher scores under Superscoring do not strictly imply higher underlying knowledge.  

\begin{proposition}[Effort Reduction and Score Bias] 
\label{prop:scorebias}
Let $e^R$ denote total expected effort across both rounds under scoring rule $R$. Under Superscoring with $N_i\ge 2$ attempts, due to the Safety Net Effect (Proposition \ref{prop:safetynet}) and Effort Specialization (Proposition \ref{prop:specialization}), expected total effort falls relative to Single-Sitting whenever the expected mechanical score boost exceeds the expected score loss from reduced effort:
\begin{align*} 
    \mathbb{E}[e^{SC}] < \mathbb{E}[e^{SS}], 
\end{align*}
 \end{proposition}
 The proof is provided in Appendix~\ref{subsec:propscorebias}. Combining Proposition \ref{prop:screening} and \ref{prop:scorebias} reveals a distortion: under Superscoring, expected submitted scores rise even as total expected learning falls. This introduces a systematic upward bias where the observed score artificially overstates true ability, driven entirely by harvesting favorable noise across multiple attempts. Because the ability to secure another attempt is constrained by the wealth-dependent fixed cost $C(W)$, this algorithmic inflation acts as a structural advantage that disproportionately benefits wealthy applicants. We formalize this result in Appendix~\ref{app:corsbr} by introducing the Score Bias Ratio (SBR), which measures this upward bias relative to the expected score under Single-Sitting.

\paragraph{Wealth Amplification}
The distortions above arise mechanically from the scoring rule, but the opportunity to exploit them is not equally distributed. Wealth does not directly increase latent ability; it reduces the effective cost $C(W_i)$ of additional attempts, entering the admissions process indirectly through strategic score generation.

\begin{proposition}[Access Premium]
\label{prop:access}
Consider two students with identical latent ability $(\eta_M,\eta_V)$. Student A is unconstrained $(C^A<\infty)$ and optimally retakes. Student B faces prohibitive barriers $(C^B\to\infty)$, restricting them to one attempt. Under linear effort costs and non-degenerate noise, Student A achieves a strictly higher expected final score:
\begin{align*}
            \mathbb{E}[S^A] > \mathbb{E}[S^B].
\end{align*}
We call the difference $\mathbb{E}[S^A]-\mathbb{E}[S^B]$ the ``Access Premium''.
\end{proposition}
(Proof in Appendix~\ref{subsec:propaccess}.)
Because Student A has the resources to take multiple attempts, their feasible choice set strictly contains Student B's. Superscoring therefore rewards both academic ability and access to repeated testing. As formalized in Appendix~\ref{app:abilitygap}, closing this Access Premium requires the constrained student to possess strictly higher latent ability to overcome their lack of retake opportunities.


\begin{proposition}[Superscoring as an Inequality Accelerant]
\label{prop:accelerant}
Suppose two demographic groups $H$ (high-wealth) and $L$ (low-wealth) differ only in effective retake costs, yielding attempt frequencies $N_H > N_L \ge 1$. 
Under Superscoring, the expected submitted-score gap between high- and low-wealth groups strictly exceeds the gap under Single-Sitting:
\begin{align*}
\mathbb{E}[S_H^{SC}] - \mathbb{E}[S_L^{SC}] > \mathbb{E}[S_H^{SS}] - \mathbb{E}[S_L^{SS}].\end{align*}
\end{proposition}
(Proof in Appendix~\ref{subsec:propaccelerant}.)
Under Single-Sitting, the score gap primarily reflects the behavioral differences induced by heterogeneous retake costs. Under Superscoring, this gap widens because wealthier students can afford more attempts, gaining larger artificial score boosts. The admissions policy itself then amplifies existing resource differences.

\subsection{Institutional Tradeoff}
\label{subsec:institutetradeoff}
From the student's perspective, the aggregation policy distorts effort and retake behavior. From the university's perspective, these responses generate systematic score inflation and wealth-based access premiums. We first evaluate how Superscoring can positively impact student participation, then formally characterize the resulting institutional trade-off.

\paragraph{Superscoring Changes Student Participation}
We formally define participation as the decision to enter the applicant pool and take the first-round exam. A student participates ($A_i = 1$) if their ex-ante expected utility is non-negative ($V_1^R \ge 0$). If $V_1^R < 0$, increasing effort makes the student strictly worse off, and they exit the pool entirely ($A_i = 0$). Note that this initial participation decision is distinct from the decision to retake, which is governed by the Stage 2 continuation value ($V_2^R$). To study the limits of participation, we invoke strictly convex variable costs ($\alpha > 1$).

\begin{proposition}[Participation Effect]
\label{prop:participation}
Under convex multidimensional effort costs $c(x,y)=\frac{k}{\alpha}(x+y)^\alpha$ with $\alpha>1$, there exist students for whom the ex-ante expected utility is strictly positive under Superscoring but negative under Single-Sitting ($V_1^{SC} \ge 0 > V_1^{SS}$). Superscoring therefore induces participation $(A_i(SC)=1)$ for constrained students who would otherwise exit the applicant pool entirely $(A_i(SS)=0)$. This result requires $\alpha > 1$ and does not hold under linear costs.
\end{proposition}
The proof is provided in Appendix~\ref{subsec:propparticipation}. Since $(x + y)^\alpha > x^\alpha + y^\alpha$ for $\alpha>1$, simultaneous preparation costs under Single-Sitting are superadditive. Superscoring avoids this compound penalty by allowing constrained students to target one subject per attempt. Single-Sitting therefore excludes a pool of high-ability students not for lack of merit, but because they cannot afford the simultaneous effort penalty.

\paragraph{Evaluating the Policy Trade-off } 
To formalize the trade-off between expanding participation and degrading the statistical signal, we evaluate a stylized binary setting. Let $V_{uni}(R)$ denote the expected latent ability of the admitted class (the top $K$ students) under rule $R$.

Consider an applicant pool with two groups: a fraction $\gamma\in(0,1)$ are constrained Type A students with high ability $\eta_A$ who participate only under Superscoring. The remaining $(1-\gamma)$ are unconstrained Type B students with lower ability $\eta_B<\eta_A$ who participate under both policies. We assume $\eta_A>\bar\eta_B$, where $\bar\eta_B$ is the expected ability of the top-$K$ Type B students under Single-Sitting.

\begin{proposition}[Institutional Trade-Off]
\label{prop:tradeoff}
Let $V_{uni}(R)$ denote the expected true ability of the top $K$ admitted students under scoring rule $R$. Under Single-Sitting, because constrained Type A students exit the pool, the admitted class is drawn entirely from Type B applicants: $V_{uni}(SS)=\bar\eta_B$. Under Superscoring, both groups participate, and the expected admitted class contains a probabilistic mixture of high-ability Type A students and artificially inflated Type B students. Because $V_{uni}(SC,\gamma)$ is strictly increasing and continuous in $\gamma$, there exists a unique threshold $\gamma^*\in(0,1)$ such that:
\begin{align*}
            V_{uni}(SC) > V_{uni}(SS) \iff \gamma > \gamma^*.
\end{align*}
\end{proposition}
The proof is provided in Appendix~\ref{subsec:proptradeoff}. Superscoring is preferred by the university if and only if the recovered constrained population is large enough to mathematically offset bias introduced among the unconstrained population.


\section{Experimental Results and Corrections}
\label{sec:experiments_and_corrections}
Neither existing rule simultaneously achieves accurate measurement and broad participation. While Superscoring preserves applicant participation, it introduces noise-driven score inflation. To retain this participation benefit while neutralizing inflation, we propose three algorithmic interventions that correct for strategic gaming under varying information constraints. Specifically, we introduce the Variance Tax and the Full-Information Score Correction (FISC) as structural scoring rules, while the Contextual Correction acts as a post-processing intervention applied directly to preserve the benefits of Superscoring.

\subsection{Alternative Scoring Rule Design} 
\label{subsec:alternative_rules}
To mitigate the distortions of uncorrected Superscoring, we propose three alternative mechanisms tailored to different institutional information constraints. We distinguish these rules by whether they act as structural scoring rules—which change how students study and retake the exam—or as post-processing interventions applied after scores are submitted.

\paragraph{FISC: Full-Information Score Correction} If the university observes a student's exact attempt count $N_i$, it can apply an additive correction of the form $S_i^{FISC} = S_i^{SC} - g(N_i)$. In Appendix~\ref{app:mechanisms}, we characterize the optimal Information-Constrained Correction (Theorem~\ref{thm:optimal}). We prove that the mean-squared-error minimizing penalty is the expected noise premium: $g^*(N_i)=\mathbb{E}[\Omega(N_i)\mid N_i]$. When deployed as a structural rule, FISC changes student behavior. By taxing multiple attempts, it deters students from retaking the exam to harvest noise, while retaining the max operator's participation benefits to correct a deficient subject.

\paragraph{TAX: Variance Tax} If the university mandates the submission of all attempts, it can penalize variance directly through a convex combination of scores: $\widehat S_{VT}(\lambda) = \lambda\max(S_1,S_2)+(1-\lambda)\min(S_1,S_2)$. As shown in Appendix~\ref{app:vartax}, setting $\lambda \in (0.5, 1)$ compresses the right tail of the score distribution without requiring the university to estimate the expected noise premium. Like FISC, the Variance Tax is a structural rule that reduces the expected value of retaking solely to harvest noise, endogenously shifting student effort.

\paragraph{CTX: Contextual Correction} Because demographic-blind statistical corrections can exclude high-ability, low-wealth applicants from the right tail, the Contextual Correction instead targets Equal Opportunity. If the university observes applicant wealth cohorts $q\in\mathcal{G}$, it applies a group-specific penalty $c^*(q)$ to equalize False Negative Rates across wealth groups for equally qualified applicants (Appendix~\ref{app:mechanisms}). Unlike FISC and the Variance Tax, the Contextual Correction is a post-processing intervention. Because the penalty $c^*(q)$ is constant for a given demographic group, it drops out of the student's marginal utility maximization. Consequently, it preserves individual effort allocations; students optimize exactly as they would under standard Superscoring.

\subsection{Experimental Results}\label{subsec:results} 
We computationally simulate the student dynamic program over $M=200$ Monte Carlo draws for $N=10,000$ students. For the full experimental setup, including the strictly convex effort costs ($\alpha=2$) and exact parameter calibration to the \textit{2025 Total Group SAT Suite of Assessments Annual Report} \cite{collegeboard2025}, refer to Appendix~\ref{app:experimental_setup}. Our simulation assumes students optimize their effort and retake decisions against the active structural scoring rule (Single-Sitting, Superscoring, or the Variance Tax). Because the Contextual Correction is a post-processing mechanism, it mathematically inherits the broad participation and effort distributions generated under Superscoring.


We now present simulation results for the empirically calibrated 2025 applicant pool, where students face wealth-dependent effort costs. To isolate the effects of the aggregation rules from the empirical skew of the 2025 demographics, we also run controlled baseline simulations using a perfectly balanced population (where each wealth quintile is about 20\% of the applicant pool) in Appendix~\ref{app:uniformpop}.

\paragraph{Endogenous Effort and Participation}
Figure \ref{fig:effort_decomp_empirical} demonstrates that there is a significant contrast in effort capacity across wealth quintiles. Financial and temporal frictions force low-wealth students (Q1) to exert lower effort in both preparation rounds---especially the second round---compared to their wealthier peers (Q5) regardless of the scoring rule. Figure \ref{fig:retake_rates_empirical} further demonstrates that the impact is also seen in terms of retake attempts, where financial and temporal costs prevent low-income students from fully utilizing the mechanism: Q5 students retake at a 95.3\% rate under Superscoring, while Q1 students retake at only 47.4\%.



\begin{figure}[ht] 
\centering 
\includegraphics[width=\columnwidth]{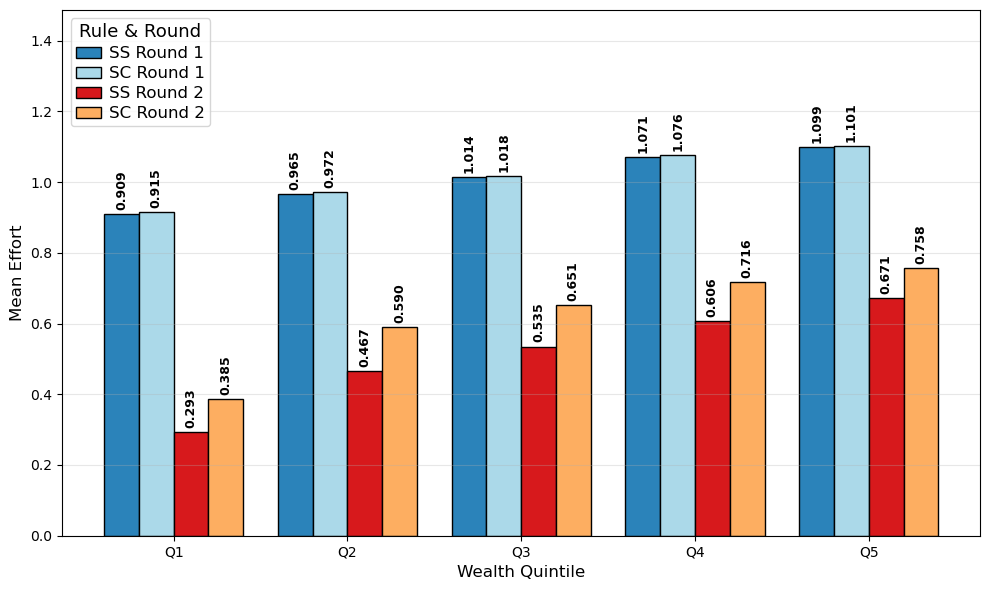}  
\caption{\footnotesize Effort Exerted by Round and Wealth Quintile under Single-Sitting in Round 1 (blue, with SS on the left and SC on the right) and Round 2 (red, with SS on the left and SC on the right).}
\label{fig:effort_decomp_empirical} 
\vspace{-3mm}
\end{figure}

\begin{figure}[ht] 
\centering 
\includegraphics[width=0.9\columnwidth]{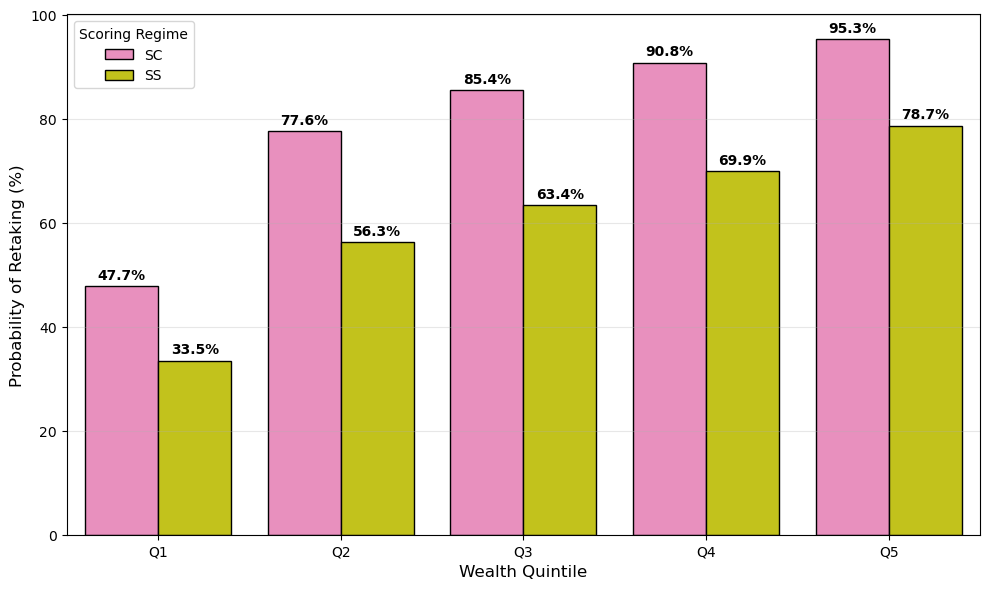} 
\caption{\footnotesize $\%$ of students who retake under Superscoring (pink) vs under Single-Sitting (yellow) for varying wealth quintiles.
}
\label{fig:retake_rates_empirical} 
\end{figure} 
\paragraph{Score Inflation and Algorithmic Premiums}
This algorithmic advantage for higher wealth populations persists even when holding latent ability constant. Figure \ref{fig:access_premium} illustrates the Access Premium across true ability quintiles. We see that students with above-median wealth extract higher final scores than their equally qualified peers with below-median wealth.
\begin{figure}[ht]
\centering
\includegraphics[width=0.9\columnwidth]{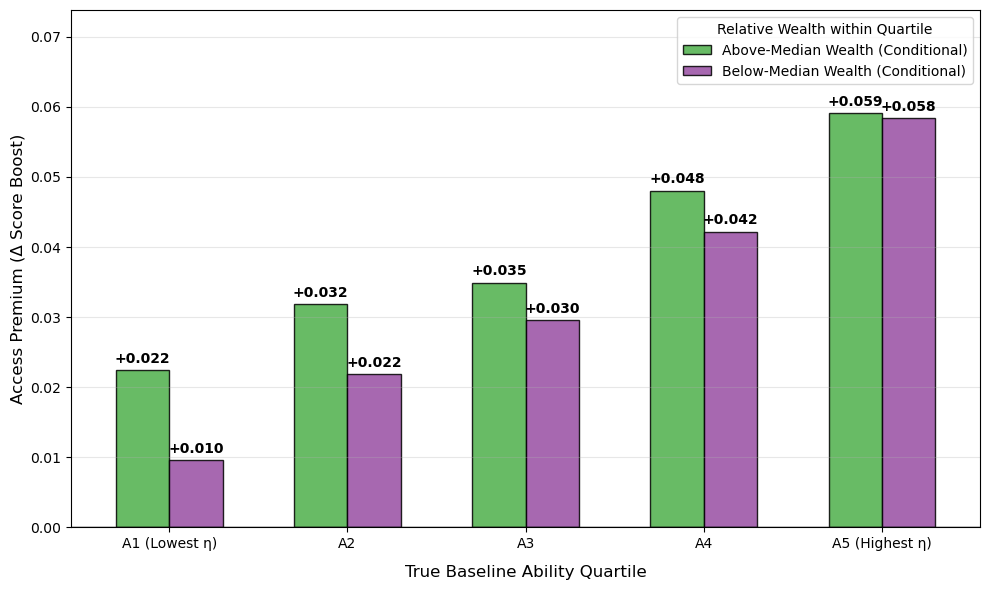} 
    \caption{\footnotesize Conditionally Balanced Access Premium by True Ability Quintile for students with above-median wealth in green, and students with below-median wealth in purple.}
\label{fig:access_premium}
\end{figure}

\paragraph{Algorithmic Correction and Institutional Fairness}
Figure \ref{fig:algorithmic_premium} further studies the algorithmic score premium and the impact of our interventions on said premium. While uncorrected Superscoring inflates scores across the board, FISC and the Variance Tax penalize this unearned variance, impacting the highest wealth quintiles more aggressively. Under these corrections, lower wealth quintiles obtain a better access premium as opposed to Superscoring. 
\begin{figure}[ht]
\centering
\includegraphics[width=0.9\columnwidth]{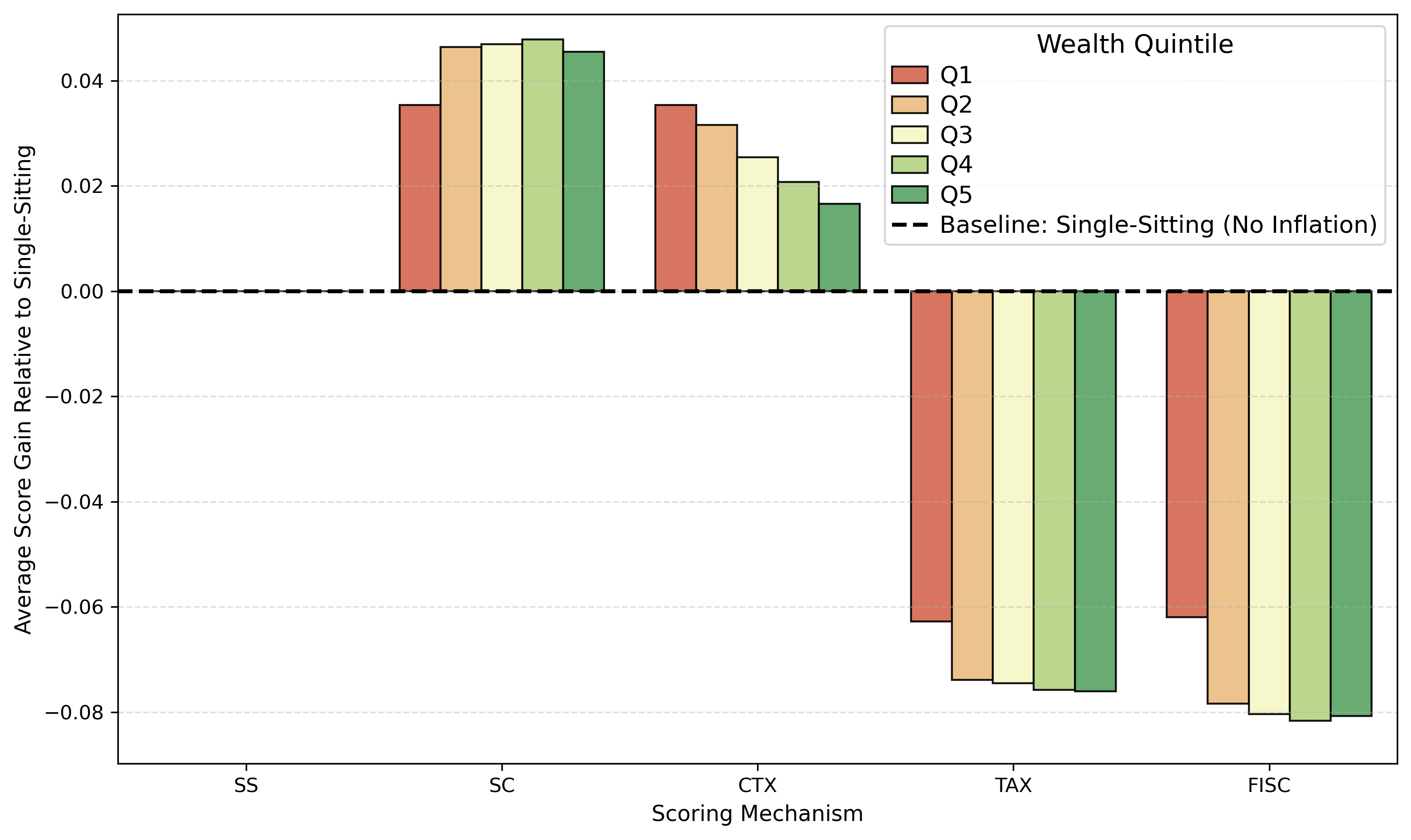} 
\caption{\footnotesize Algorithmic Score Premium relative to Single-Sitting. For each scoring rule (SS, SC, CTX, TAX, FISC), the grouped bars represent the five wealth quintiles sequentially, from Q1 to Q5.}
\label{fig:algorithmic_premium}
\end{figure}
The institutional consequences of uncorrected Superscoring are severe. Table \ref{tab:error_rates_fnr} shows how the sheer volume of Q5 applicants harvesting noise drives the aggregate admission threshold upward, pushing the False Negative Rate (FNR) for qualified Q1 students to 76.0\% (Table \ref{tab:error_rates_fnr}). The insights of Table \ref{tab:error_rates_fnr} are shown to be robust to changes in parameter calibration in Appendix~\ref{app:grid_robustness}. Finally, we note here that while False Positive Rates (FPR) also vary by policy (see Appendix~\ref{subsec:fpr}), they remain structurally low across all groups; thus, we focus on FNR as the main differentiator between our rules. 
\begin{table}[h]
\centering
\begin{tabular}{lccccc}
\toprule
\textbf{Rule} & \textbf{Q1} & \textbf{Q2} & \textbf{Q3} & \textbf{Q4} & \textbf{Q5} \\
\midrule
SS         & 76.0\%&   24.5\%&   12.3\%&    11.0\%& 8.2\%\\
SC         & 76.0\%& 28.3\%& 14.8\%& 9.2\%& 8.6\%\\
TAX        & 64.0\%& 13.2\%& 14.8\%& 9.2\%& 5.6\%\\
FISC       & 52.0\%& 18.9\%& 16.0\%& 11.0\%& 9.5\%\\
CTX& 16.0\%& 11.3\%& 12.3\%& 12.8\%& 15.1\%\\
\bottomrule
\end{tabular}
\caption{\footnotesize False Negative Rate (FNR) by scoring rule for each wealth quintile. The rules are labeled as Single-Sitting (SS), Superscoring (SC), Variance Tax (TAX), Full-Information Score Correction (FISC), and Contextual Correction (CTX).}\label{tab:error_rates_fnr}
\vspace{-3mm}
\end{table}

The Variance Tax (TAX) penalizes unearned variance, reducing the Q1 FNR to 64.0\%. However, TAX is strictly a measurement correction; it remains blind to underlying structural inequalities in effort costs. Consequently, it only returns admission rates to near the Single-Sitting baseline—leaving Q1 with 1.8\% of the admitted class, less than half of their true share of top latent ability (5.0\%) (Table~\ref{tab:admitted_demographics}). Similarly, by explicitly taxing multiple attempts, FISC targets the noise premium directly. Because students endogenously respond to this structural penalty, FISC successfully deters speculative retaking, further dropping the Q1 FNR to 52.0\%. However, while FISC neutralizes the algorithmic advantage of noise harvesting, it cannot resolve the unequal financial and temporal barriers to initial preparation. As a result, demographic-blind measurement corrections hit a structural ceiling: under FISC, Q1 remains underrepresented (2.8\% admitted versus 5.0\% True Ability), while the Q5 cohort remains overrepresented (51.0\% admitted versus 46.4\% True Ability). 

\begin{table}[h]
\centering
\begin{tabular}{lcccccc}
\toprule
\textbf{Q}& \textbf{TA}& \textbf{SS} & \textbf{SC} & \textbf{TAX} & \textbf{FISC} & \textbf{CTX}\\
\midrule
Q1& 5.0\%& 1.2\%& 1.4\%& 1.8\%& 2.8\%& 6.0\%\\
Q2          & 10.6\%& 8.0\%& 7.8\%& 9.2\%& 8.6\%& 10.6\%\\
Q3          & 16.2\%& 14.6\%& 14.2\%& 14.4\%& 14.4\%& 14.8\%\\
Q4          & 21.8\%& 23.4\%& 24.2\%& 22.0\%& 23.2\%& 22.2\%\\
Q5& 46.4\%& 52.8\%& 52.4\%& 52.6\%& 51.0\%& 46.4\%\\
\bottomrule
\end{tabular}
\caption{\footnotesize Admitted Class Composition for the top 5\% of applicants. Columns denote Wealth Quintile (Q), True Ability baseline (TA) for SS, SC, TAX, FISC, and CTX.}
\label{tab:admitted_demographics}
\end{table}

CTX, by contrast, utilizes relative structural evaluation over absolute measurement. By explicitly targeting Equal Opportunity, it sets group-specific penalties based on the conditional score CDF, evaluating Q1 students against the constraints of their own socioeconomic environment rather than the globally advantaged Q5 population. CTX successfully stabilizes error rates across the population, compressing the FNR spread to a tight, equitable band (11.3\% to 16.0\%). Because Q5 constitutes 46.4\% of the true top 5\% in latent ability, a sufficient volume of high-ability Q5 students clear the adjusted cutoff to maintain their proportional share. CTX is the only mechanism that actively redistributes admitted seats to approximate the True Ability distribution for each wealth quintile, successfully increasing Q1's share of the admitted class from 1.4\% under SC to 6.0\%.

\section*{Acknowledgements}
The authors acknowledge support from the US NSF under grants IIS-2504990 and IIS-2336236. Any opinions and findings expressed in this material are those of the authors and do not reflect the views of their funding agencies.

\clearpage
\bibliography{references}

\clearpage
\appendix

\section{Limitations and Future Work}
\label{app:futurework}
Our framework relies on several assumptions to isolate the mechanics of score aggregation. First, while we capture disparities in test preparation through wealth-dependent effort costs, we assume the distribution of random test-day noise is identical for all applicants. In practice, students who can afford multiple attempts may intentionally use high-variance guessing strategies, knowing they have a safety net. If noise variance is endogenous to wealth, the algorithmic distortion we identify could be even more pronounced. Additionally, we use continuous, deterministic effort cost functions. Future empirical work could calibrate these parameters to quantify actual wealth-to-effort translation in specific populations.

Finally, our current model treats university admission as a terminal payoff, abstracting away from the downstream consequences of inflated scores. A natural direction for future work is embedding this strategic behavior into broader multi-stage selection pipelines. In practice, admission is often the first step in a sequence of evaluations, generating transcripts that employers later use for hiring. When applicants with different initial wealth budgets strategically allocate effort to pass a sequence of linear classifiers, longer screening pipelines might compound the initial inequalities of Superscoring. Alternatively, downstream signals might eventually correct them. Extending our framework to these sequential environments would help clarify how early-stage evaluation rules impact long-term fairness.

\section{Proofs from Section~\ref{sec:student_behavior}: Characterization of Strategic Student Behavior}
\subsection{Optimal Stopping and Target Proficiency} 
\label{app:target_lemma}
To prove the behavioral characteristics in Section~\ref{sec:student_behavior}, we first establish that under linear costs, students target a unique, deterministic proficiency level. Because test-day noise $\delta$ is additive and mean-zero ($\mathbb{E}[\delta]=0$), it does not affect the marginal return to effort: $\frac{\partial}{\partial e}\mathbb{E}[S_{final}] = \frac{\partial}{\partial e}f(E_j)$. The optimal effort choice therefore depends only on a unique deterministic proficiency level $E_j^*$.

\begin{lemma}[Target Proficiency Levels]
\label{lem:target}
Let $E_M=\eta_M(x_1+x_2)$ and $E_V=\eta_V(y_1+y_2)$ denote total effort in Mathematics and Verbal respectively. Under constant marginal cost $k>0$, conditional on exerting positive Stage~2 effort in subject $j\in\{M,V\}$, the unique target proficiency level $E_j^*$ satisfies:
\[
    f'(E_j^*)=\frac{\beta k}{w_j \eta_j}.
\]
The expected score at this proficiency level is $f(E_j^*)$. We denote the expected scores at these optimal targets as $X^* = f(E_M^*)$ and $Y^* = f(E_V^*)$.
\end{lemma}

\begin{proof}
Let $j \in \{M, V\}$ index the subject. To map to our model primitives, let $e_{1j} \in \{x_1, y_1\}$ and $e_{2j} \in \{x_2, y_2\}$ represent the effort exerted in Round 1 and Round 2 for subject $j$, respectively.

Consider the Stage 2 continuation problem evaluated at the expected second-round shock ($\mathbb{E}[\delta_2] = 0$). Under linear costs ($\alpha=1$), the variable cost function is $c(e_2) = k \cdot e_2$, where $k>0$ is the constant marginal cost.

Assume the student is at an interior solution for subject $j$ ($e_{2j} > 0$), meaning additional effort strictly increases the expected submitted score. The expected deterministic score component in Stage 2 is $w_j f(\eta_j(e_{1j} + e_{2j}))$. Taking the derivative with respect to Stage 2 effort and incorporating the discount factor $\beta$, the First-Order Condition (FOC) requires the marginal expected benefit to equal the marginal cost: 
$$w_j \eta_j f'(\eta_j(e_{1j} + e_{2j})) = \beta k$$

Substituting the definition of cumulative knowledge $E_j = \eta_j(e_{1j} + e_{2j})$:
$$w_j \eta_j f'(E_j) = \beta k \iff f'(E_j) = \frac{\beta k}{w_j \eta_j}$$
    
By assumption, the knowledge production function $f(\cdot)$ is strictly concave ($f'' < 0$). Therefore, its first derivative $f'(\cdot)$ is strictly monotonically decreasing. This guarantees that the inverse function $(f')^{-1}$ exists and is unique.

Assumption~\ref{ass:viable} ensures that the marginal cost lies within the range of the marginal benefit curve, meaning $\frac{\beta k}{w_j \eta_j} \le f'(0)$. We can therefore explicitly solve for the optimal target knowledge level: 
$$E_j^* = (f')^{-1}\left( \frac{\beta k}{w_j \eta_j} \right)$$
Because $f(\cdot)$, $k$, $\beta$, $w_j$, and $\eta_j$ are all exogenous parameters, $E_j^*$ is a unique, fixed constant. It depends solely on the student's type and cost parameters, making it independent of first-round score realizations or test-day noise.

Finally, the optimal second-round effort required to reach this target is the difference between the target proficiency and the effective proficiency already accumulated in Round 1, scaled by ability: 
$$e_{2j}^* = \frac{E_j^* - \eta_j e_{1j}}{\eta_j} = \frac{E_j^*}{\eta_j} - e_{1j}$$

This interior solution holds when Round 1 proficiency is strictly below the target ($\eta_j e_{1j} < E_j^*$). If initial effective effort already meets or exceeds the target ($\eta_j e_{1j} \ge E_j^*$), the non-negativity constraint binds, and $e_{2j}^* = 0$.
\end{proof}

\subsection{Second-Order Conditions and Concavity}
\label{app:concavity}
Before proving Proposition~\ref{prop:monotone}, we first establish a necessary structural property regarding the concavity of the student's optimization problem. This guarantees that the first-order conditions yield a unique interior maximum.

\begin{lemma}[Concavity of Stage 1 Objective]
\label{lem:concavity}
Under the continuous noise model, the Stage 1 objective $V_1^R(x_1, y_1)$ is strictly concave at any interior optimum, provided the noise distribution is sufficiently dispersed relative to the curvature of the knowledge production function $f(\cdot)$. Specifically, we maintain the parameter restriction that the strict concavity of the expected scores dominates the marginal boundary effects of the retake region, ensuring the Hessian of $V_1^R$ is negative definite.
\end{lemma}

\begin{proof}
By applying the Leibniz Integral Rule to the Stage 1 expected utility (as detailed in Appendix~\ref{subsec:proofmonotone}), the first-order condition for Round 1 Math effort isolates the marginal return over the stopping and retaking regions:
\begin{align*}
    \frac{\partial V_1^R}{\partial x_1} = \mathbb{E}_{\boldsymbol{\delta}_1} \left[ \frac{\partial S_1^R}{\partial x_1} \mathbf{1}_{\{\boldsymbol{\delta}_1 \notin \mathcal{R}\}} + \frac{\partial V_2^R}{\partial x_1} \mathbf{1}_{\{\boldsymbol{\delta}_1 \in \mathcal{R}\}} \right] - c'(x_1+y_1 \mid W_i) = 0
\end{align*}

Taking the second derivative with respect to $x_1$ yields two types of terms: the expected second derivatives of the score functions within each region, and a boundary term representing the marginal change in the retake probability density.

\begin{align*}
    \frac{\partial^2 V_1^R}{\partial x_1^2} = \mathbb{E}_{\boldsymbol{\delta}_1} \left[ \frac{\partial^2 S_1^R}{\partial x_1^2} \mathbf{1}_{\{\notin \mathcal{R}\}} + \frac{\partial^2 V_2^R}{\partial x_1^2} \mathbf{1}_{\{\in \mathcal{R}\}} \right] 
    + \text{BoundaryEffect}(x_1)
    - c''(x_1+y_1 \mid W_i)
\end{align*}

Because the knowledge production function $f(\cdot)$ is strictly concave ($f'' < 0$), both $\frac{\partial^2 S_1^R}{\partial x_1^2}$ and $\frac{\partial^2 V_2^R}{\partial x_1^2}$ are strictly negative. Furthermore, because the cost function is convex for all  ,  , meaning the subtracted cost term further pulls the second derivative negative.

While the exact sign of the Boundary Effect depends on the local slope of the noise density function $\phi(\cdot)$, the expectation terms are strictly negative and scale with the curvature $f''(\eta_M x_1)$. As long as the variance of the noise $\sigma^2$ is sufficiently large (meaning the density $\phi$ is dispersed and does not spike drastically at the boundary), the strictly negative structural curvature dominates: $\frac{\partial^2 V_1^R}{\partial x_1^2} < 0$. Symmetrically, $\frac{\partial^2 V_1^R}{\partial y_1^2} < 0$.

For joint concavity, the Hessian matrix $H$ must be negative definite. This requires the determinant to be strictly positive:
$$\left(\frac{\partial^2 V_1^R}{\partial x_1^2}\right) \left(\frac{\partial^2 V_1^R}{\partial y_1^2}\right) - \left(\frac{\partial^2 V_1^R}{\partial x_1 \partial y_1}\right)^2 > 0$$
Because the cross-partial derivative $\frac{\partial^2 V_1^R}{\partial x_1 \partial y_1}$ is driven entirely by the joint boundary density of the independent shocks, we maintain the standard sufficient parameter restriction that the noise distribution is sufficiently dispersed relative to $f''$. This restriction guarantees that the product of the direct concavities strictly dominates the squared cross-partial effect, yielding a negative definite Hessian and guaranteeing a unique joint interior optimum.
\end{proof}

\subsection{Proposition~\ref{prop:monotone}: Monotonicity in Fixed Retake Costs}
\label{subsec:proofmonotone}
\begin{proof}
We analyze the student's Stage 1 optimization problem. The student chooses initial effort $x_1$ and $y_1$ to maximize total expected utility:
\begin{align*}
    V_1^R(x_1, y_1) = \mathbb{E}_{\delta_1}\Big[ \max \big( S_1^R(x_1,y_1,\delta_1), V_2^R(x_1,y_1, C) \big) \Big] 
    - c(x_1+y_1 \mid W_i),
\end{align*}
where $S_1^R$ is the realized Round 1 score and $V_2^R$ is the Stage 2 continuation value. A student optimally retakes if and only if $V_2^R > S_1^R$. 

Let $\boldsymbol{\delta}_1 = (\delta_{M1}, \delta_{V1})$ represent the vector of Round 1 shocks. We define the Retake Region $\mathcal{R}(C) = \{ \boldsymbol{\delta}_1 : V_2(x_1, y_1, C) > S_1(x_1, y_1, \boldsymbol{\delta}_1) \}$. Using an indicator function $\mathbf{1}_{\{\cdot\}}$, we can rewrite the Stage 1 expected utility as:
\begin{align*}
    V_1^R(x_1, y_1)
    = \mathbb{E}_{\boldsymbol{\delta}_1} \Big[
        S_1^R(x_1, y_1, \delta_1)\,
        \mathbf{1}_{\{\boldsymbol{\delta}_1 \notin \mathcal{R}\}}
        + V_2^R(x_1, y_1, C)\,
        \mathbf{1}_{\{\boldsymbol{\delta}_1 \in \mathcal{R}\}}
        \Big] 
    - c(x_1 + y_1 \mid W_i).
\end{align*}

Assuming an interior solution, the optimal first-round effort $x_1^*$ must satisfy the first-order condition $\frac{\partial V_1^R}{\partial x_1} = 0$. By the Envelope Theorem for continuous random variables, the derivative of the boundary between the regions evaluates to zero (since $S_1^R = V_2^R$ exactly on the boundary). We can pass the derivative inside the expectation:
\begin{align*}
    \frac{\partial V_1^R}{\partial x_1} =  \mathbb{E}_{\boldsymbol{\delta}_1} \left[ \frac{\partial S_1^R}{\partial x_1} \mathbf{1}_{\{\boldsymbol{\delta}_1 \notin \mathcal{R}\}} + \frac{\partial V_2^R}{\partial x_1} \mathbf{1}_{\{\boldsymbol{\delta}_1 \in \mathcal{R}\}} \right] 
    - c'(x_1+y_1 \mid W_i) = 0.
\end{align*}

By adding and subtracting $\frac{\partial S_1^R}{\partial x_1}\mathbf{1}_{\{\boldsymbol{\delta}_1 \in \mathcal{R}\}}$ inside the expectation, we rearrange this to isolate the impact of the retake option:
\begin{align*}
    \frac{\partial V_1^R}{\partial x_1} = \mathbb{E}_{\boldsymbol{\delta}_1} \left[ \frac{\partial S_1^R}{\partial x_1} \right] + \mathbb{E}_{\boldsymbol{\delta}_1} \left[ \left(\frac{\partial V_2^R}{\partial x_1} - \frac{\partial S_1^R}{\partial x_1}\right) \mathbf{1}_{\{\boldsymbol{\delta}_1 \in \mathcal{R}(C)\}} \right]
    - c'(x_1+y_1 \mid W_i) = 0.
\end{align*}

To determine how optimal effort $x_1^*$ changes with respect to the fixed retake cost $C$, we evaluate the cross-partial derivative $\frac{\partial^2 V_1^R}{\partial x_1 \partial C}$. 

Because the fixed cost $C$ is paid only upon retaking, it enters the continuation value purely as an additive constant: $\frac{\partial V_2^R}{\partial C} = -\beta$. Consequently, the integrands $\frac{\partial V_2^R}{\partial x_1}$ and $\frac{\partial S_1^R}{\partial x_1}$ are independent of $C$. The only term dependent on $C$ is the boundary of the retake region $\mathcal{R}(C)$. Because an increase in $C$ strictly decreases $V_2^R$ while leaving $S_1^R$ unchanged, the retake condition $V_2^R > S_1^R$ becomes harder to satisfy, causing the region $\mathcal{R}(C)$ to strictly shrink.

We must now evaluate the sign of the integrand over this shrinking region. Because the knowledge production function $f(\cdot)$ is strictly concave ($f'' < 0$), the marginal return to first-round effort strictly decreases as cumulative effective effort increases. Because $V_2^R$ anticipates optimal Round 2 effort $x_2^* \ge 0$, while $S_1^R$ evaluates effort strictly at $x_1$, the expected marginal return is strictly lower under continuation: $\frac{\partial V_2^R}{\partial x_1} < \frac{\partial S_1^R}{\partial x_1}$.

Therefore, the term $\left(\frac{\partial V_2^R}{\partial x_1} - \frac{\partial S_1^R}{\partial x_1}\right)$ is strictly negative. Because the expectation is integrating a strictly negative quantity over the region $\mathcal{R}(C)$, and an increase in $C$ strictly shrinks this region, the integral becomes less negative. This results in a strictly positive cross-partial derivative:
$$\frac{\partial^2 V_1^R}{\partial x_1 \partial C} > 0.$$

The positive cross-partial implies that the objective function exhibits strictly increasing differences in $(x_1, C)$. Because the Stage 1 objective is strictly concave by Lemma~\ref{lem:concavity} ($\frac{\partial^2 V_1^R}{\partial x_1^2} < 0$), applying the Implicit Function Theorem yields:
$$\frac{\partial x_1^*}{\partial C} = - \frac{\frac{\partial^2 V_1^R}{\partial x_1 \partial C}}{\frac{\partial^2 V_1^R}{\partial x_1^2}} > 0.$$
By symmetric logic, differentiating with respect to $y_1$ yields $\frac{\partial^2 V_1^R}{\partial y_1 \partial C} > 0$, guaranteeing $\frac{\partial y_1^*}{\partial C} > 0$.
\end{proof}

\subsection{Proposition~\ref{prop:specialization}: Effort Specialization across Rules}
\label{subsec:proofspecialization}
\begin{proof}
We evaluate the optimal Round 2 effort allocation $(x_2^*, y_2^*)$ under the certainty-equivalent continuation problem ($\mathbb{E}[\delta_2] = 0$). By Assumption~\ref{ass:costs}, under linear costs ($\alpha=1$), the variable cost function is $c(e_2) = k \cdot e_2$, where $k>0$ is the constant marginal cost. Assume the student elects to retake, meaning the expected continuation value strictly exceeds the stopping value.

By Lemma~\ref{lem:target}, the optimal continuation target in subject $j \in \{M,V\}$ is defined by the interior condition $E_j^* = (f')^{-1}\left( \frac{\beta k}{w_j \eta_j} \right)$. We define the corresponding deterministic target score components as $X^* \equiv f(E_M^*)$ and $Y^* \equiv f(E_V^*)$.

\textbf{Scoring Rule 1: Single-Sitting (SS)}
Under Single-Sitting, the continuation objective evaluates the new sitting holistically:
\begin{align*}
    \max_{x_2, y_2 \ge 0} \Pi_2^{SS}(x_2, y_2)  = w_M f(\eta_M(x_1+x_2))
    + w_V f(\eta_V(y_1+y_2))
    - \beta\big[k(x_2+y_2) + C\big].
\end{align*}

The marginal returns for $x_2$ and $y_2$ depend only on cumulative effective effort, entirely independent of the realized first-round scores $(X_1, Y_1)$ and their associated noise $(\delta_{M1}, \delta_{V1})$. 

Provided the realized score was generated by sub-target knowledge ($\eta_M x_1 < E_M^*$ and $\eta_V y_1 < E_V^*$), the initial effort falls strictly short of the interior optimum targets derived in Lemma~\ref{lem:target}. To satisfy the first-order conditions, the student must strictly choose positive effort to close the gap: $x_2^* = \frac{E_M^*}{\eta_M} - x_1 > 0$ and $y_2^* = \frac{E_V^*}{\eta_V} - y_1 > 0$. The student actively studies for both dimensions, regardless of the noise realization in Round 1.

\textbf{Scoring Rule 2: Superscoring (SC)}
Under Superscoring, the continuation objective is piecewise because the student retains the maximum subject score across attempts: 
\begin{align*}
    \max_{x_2, y_2 \ge 0} \Pi_2^{SC}(x_2, y_2) = w_M \max\{X_1, f(\eta_M(x_1+x_2))\}
    + w_V \max\{Y_1, f(\eta_V(y_1+y_2))\}
    - \beta\big[k(x_2+y_2) + C\big].
\end{align*}

Assume without loss of generality that the banked Math score clears the deterministic target: $X_1 \ge X^* \equiv f(E_M^*)$.

We evaluate the marginal return for Math effort ($x_2$). The objective function $\max\{X_1, f(\eta_M(x_1+x_2))\}$ is continuous but exhibits a kink where the new deterministic score exactly equals the banked score. 

For any effort level where the new score is strictly less than the banked score ($f(\eta_M(x_1+x_2)) < X_1$), the marginal benefit of $x_2$ is exactly $0$, which is strictly less than the marginal cost $\beta k$. 

For any effort level that meets or exceeds the banked score ($f(\eta_M(x_1+x_2)) \ge X_1$), it must mathematically hold that $f(\eta_M(x_1+x_2)) \ge f(E_M^*)$. Because $f$ is strictly increasing, this implies $\eta_M(x_1+x_2) \ge E_M^*$. By the strict concavity of the knowledge production function ($f'' < 0$), the marginal benefit is monotonically decreasing, yielding:
$$w_M \eta_M f'(\eta_M(x_1+x_2)) \le w_M \eta_M f'(E_M^*) = \beta k.$$
Thus, the right-hand derivative (the marginal benefit of strictly increasing effort) is bounded above by $\beta k$. Because the marginal cost of effort is exactly $\beta k$, the net marginal payoff is non-positive at the boundary, and strictly negative for all subsequent $x_2$. 

Because the net marginal return to effort is strictly negative for any $x_2 > 0$, the Kuhn-Tucker conditions dictate a strict corner solution: $x_2^* = 0$.

Finally, because the student elected to retake, they must expect a strictly positive score gain to offset the fixed cost $\beta C$. If the non-negativity constraint also bound for Verbal ($y_2^*=0$), the new portfolio would exactly equal the banked Round 1 scores. The expected continuation payoff would be $-\beta C < 0$, meaning the student would have strictly preferred to stop, contradicting the premise. Therefore, all optimizing effort is exclusively allocated to the lacking subject, where the interior condition holds: $w_V \eta_V f'(\eta_V(y_1+y_2)) = \beta k$.
\end{proof}

\subsection{Proposition~\ref{prop:safetynet}: Safety Net Effect}
\label{subsec:proofsafetynet}
\begin{proof}
We first prove part (i), pointwise dominance of the continuation value. Evaluating the continuation problem under the certainty-equivalent second-round shock ($\mathbb{E}[\delta_2] = 0$), let $V_2^{SS}(X_1, Y_1)$ and $V_2^{SC}(X_1, Y_1)$ denote the maximized expected continuation values net of costs under Single-Sitting and Superscoring, respectively.

For Single-Sitting:
\begin{align*}
    V_2^{SS}(X_1, Y_1) = \max_{x_2, y_2 \ge 0}
    \Big\{ w_M f(\eta_M(x_1+x_2))
    + w_V f(\eta_V(y_1+y_2))
    - \beta k(x_2+y_2) - \beta C \Big\}
\end{align*}
For Superscoring:
\begin{align*}
    V_2^{SC}(X_1, Y_1) = \max_{x_2, y_2 \ge 0} \Big\{ w_M \max \{X_1, f(\eta_M(x_1+x_2))\}
    + w_V \max\{Y_1, f(\eta_V(y_1+y_2))\}
    - \beta k(x_2+y_2) - \beta C \Big\}
\end{align*}

By the fundamental property of the maximum operator, $\max\{A, B\} \ge B$. Therefore, for any arbitrary choice of effort $(x_2, y_2)$ and any realized scores $X_1, Y_1$:
$$\max\{X_1, f(\eta_M(x_1+x_2))\} \ge f(\eta_M(x_1+x_2))$$
$$\max\{Y_1, f(\eta_V(y_1+y_2))\} \ge f(\eta_V(y_1+y_2))$$

Because $w_M, w_V > 0$, the objective function inside the maximization for Superscoring is pointwise greater than or equal to the objective function for Single-Sitting. Taking the maximum preserves this inequality:
$$V_2^{SC}(X_1, Y_1) \ge V_2^{SS}(X_1, Y_1).$$
Strict inequality holds precisely when the student leverages a banked score (Proposition~\ref{prop:specialization}).

We now prove part (ii), the equilibrium consequence on Stage 1 effort. The student chooses Round 1 effort $x_1$ and $y_1$ to maximize total expected utility:
$$V_1^R(x_1, y_1) = \mathbb{E}_{\boldsymbol{\delta}_1}\Big[ \max \big( S_1, V_2^R(x_1,y_1) \big) \Big] - c(x_1+y_1 \mid W_i)$$

Applying the Leibniz Integral Rule over the Retake Region $\mathcal{R}^R = \{ \boldsymbol{\delta}_1 : V_2^R > S_1 \}$, the first-order condition for Math effort $x_1$ isolates the expected marginal benefit across the stopping and retaking states:
\begin{align*}
    \frac{\partial V_1^R}{\partial x_1} =  \mathbb{E}_{\boldsymbol{\delta}_1} \left[ \frac{\partial S_1}{\partial x_1} \mathbf{1}_{\{\boldsymbol{\delta}_1 \notin \mathcal{R}^R\}} + \frac{\partial V_2^R}{\partial x_1} \mathbf{1}_{\{\boldsymbol{\delta}_1 \in \mathcal{R}^R\}} \right]
    - c'(x_1+y_1 \mid W_i) = 0.
\end{align*}

To compare the marginal benefit under the two rules, we evaluate the continuation derivative $\frac{\partial V_2^R}{\partial x_1}$ using the Envelope Theorem. 

Under Single-Sitting, if the student retakes, they always choose an interior effort allocation ($x_2^* > 0$). Because cumulative effort is additive, Round 1 effort perfectly substitutes for Round 2 effort. The marginal return to Round 1 effort is exactly the discounted marginal cost of the Round 2 effort it replaces:
$$ \frac{\partial V_2^{SS}}{\partial x_1} = \beta k. $$

Under Superscoring, the student may hit a corner solution ($x_2^* = 0$) if they bank a Round 1 score. When the non-negativity constraint binds, Round 1 effort no longer substitutes for Round 2 effort. The marginal return is strictly determined by the concave knowledge function: $\frac{\partial V_2^{SC}}{\partial x_1} = w_M \eta_M f'(\eta_M x_1)$. Crucially, Proposition~\ref{prop:specialization} establishes that the non-negativity constraint *only* binds when this marginal benefit is strictly less than the marginal cost ($\beta k$). Therefore:
$$ \frac{\partial V_2^{SC}}{\partial x_1} < \beta k \quad \text{when } x_2^* = 0. $$
Otherwise, if $x_2^* > 0$, $\frac{\partial V_2^{SC}}{\partial x_1} = \beta k$. Therefore, the marginal continuation return under Superscoring is pointwise weakly smaller: $\frac{\partial V_2^{SC}}{\partial x_1} \le \frac{\partial V_2^{SS}}{\partial x_1}$.

Furthermore, because $V_2^{SC} \ge V_2^{SS}$, the retake region expands under Superscoring ($\mathcal{R}^{SS} \subseteq \mathcal{R}^{SC}$). Because students retake more frequently under SC, they are more frequently exposed to this capped marginal return ($\beta k$ or lower), rather than the higher unbounded return of the stopping state ($\frac{\partial S_1}{\partial x_1}$).

Because the marginal benefit integrand is pointwise weakly smaller, and probability mass shifts to lower-return states, the entire expected marginal benefit integral is strictly lower under SC. By Lemma~\ref{lem:concavity}, the Stage 1 objective is strictly concave, guaranteeing that a downward shift in marginal benefit reduces optimal effort:
$$x_1^{SC} \le x_1^{SS} \quad\text{and}\quad y_1^{SC} \le y_1^{SS}.$$
Strict inequality holds when the corner solution probability is non-zero, proving the Safety Net effect.
\end{proof}

\section{Extensions for Strictly Convex Costs ($\alpha > 1$)}
\label{app:convex_robustness}
In the main text, Section~\ref{sec:student_behavior} characterizes student behavior under linear costs. Here, we extend this analysis to strictly convex specifications $c(e \mid W_i) = \frac{k(W_i)}{\alpha} e^\alpha$ for $\alpha > 1$.

\subsection{The Smoothing Motive}
The introduction of cost curvature creates a powerful ``Smoothing Motive'': a desire to spread effort evenly across subjects and time periods to avoid the strictly increasing marginal penalties that convex costs introduce.

\subsubsection{Formal Proof of Submodularity and Effort Substitution}
\label{proof:ConvexCrowdingOut}
Under linear costs, subjects act as strategic substitutes purely through the probability of retaking. We now formally prove that under strictly convex effort costs, first-round efforts $x_1$ and $y_1$ become strict strategic substitutes through a physical exhaustion channel.

\begin{proof}
The student's Stage 1 expected utility maximization problem is:
\begin{align*}
    \max_{x_1, y_1} V(x_1, y_1) = \mathbb{E}[S_{final}] - \frac{k}{\alpha}(x_1+y_1)^\alpha - C(W_i) - 
    \beta \Prob(\text{Retake})\mathbb{E}\big[C(W_i) + c(x_2^*+y_2^*,W_i) \mid \text{Retake}\big].
\end{align*}

To evaluate the interaction between Math and Verbal effort, we take the cross-partial derivative of the objective function with respect to $x_1$ and $y_1$. The first derivative of the physical cost with respect to $x_1$ is:
$$ \frac{\partial}{\partial x_1} \left( -\frac{k}{\alpha}(x_1+y_1)^\alpha \right) = -k(x_1+y_1)^{\alpha-1}. $$
Taking the derivative of this marginal cost with respect to $y_1$ yields the cross-partial: 
$$ \frac{\partial^2 V}{\partial x_1 \partial y_1} = -k(\alpha-1)(x_1+y_1)^{\alpha-2}. $$
Because $k > 0$, $(x_1+y_1) > 0$ for an interior solution, and costs are strictly convex ($\alpha > 1$), this cross-partial derivative is strictly negative:
$$ \frac{\partial^2 V}{\partial x_1 \partial y_1} < 0. $$

This negative cross-partial implies that the objective function is strictly submodular in $(x_1, y_1)$ with respect to effort costs. Mechanically, an increase in Math effort $x_1$ strictly increases the marginal cost of Verbal effort $y_1$.

To evaluate the comparative static of an exogenous increase in the Math admissions threshold ($\bar X$), let the first-order conditions for an interior optimum be:
$$ V_{x_1}(x_1^*, y_1^*; \bar X) = 0, \qquad V_{y_1}(x_1^*, y_1^*; \bar X) = 0. $$
Differentiating the Verbal FOC with respect to $\bar X$ and solving yields:
$$ \frac{\partial y_1^*}{\partial \bar X} = - \frac{V_{y_1 x_1} \frac{\partial x_1^*}{\partial \bar X} + V_{y_1 \bar X}}{V_{y_1 y_1}}. $$

Evaluating the components of this expression:
\begin{enumerate}
    \item By strict concavity of the objective, the denominator is strictly negative: $V_{y_1 y_1} < 0$.
    \item Raising the Math threshold $\bar X$ strictly increases $\mathrm{Prob}(\text{fail Math})$, raising the expected cost of retaking. By the Implicit Function Theorem on the Stage 1 FOC, this induces higher first-round effort: $\frac{\partial x_1^*}{\partial \bar X} > 0$.
    \item As established above, the physical cross-partial is strictly negative: $V_{y_1 x_1} < 0$.
    \item The direct effect of $\bar X$ on the marginal insurance value of Verbal effort through the retake probability channel is strictly negative: $V_{y_1 \bar X} < 0$.
\end{enumerate}

Adding the two negative terms in the numerator ($V_{y_1 x_1} \frac{\partial x_1^*}{\partial \bar X} < 0$ and $V_{y_1 \bar X} < 0$), the entire fraction is strictly positive. The leading negative sign makes the overall comparative static strictly negative:
$$ \frac{\partial y_1^*}{\partial \bar X} < 0. $$

Therefore, under convex costs, the crowding-out effect operates through two distinct, reinforcing channels: a strategic insurance channel ($V_{y_1 \bar X} < 0$, present under both linear and convex costs), and a physical exhaustion channel ($V_{y_1 x_1} < 0$, absent under linear costs), where forcing a student to clear a higher Math threshold mechanically drives down optimal Verbal preparation.
\end{proof}

\subsection{Generalization of Structural Properties}
The introduction of cost curvature creates a powerful ``Smoothing Motive'': a desire to spread effort evenly across subjects and time periods to avoid the increasing marginal penalties that convex costs introduce. Despite this smoothing motive, the high-level structural properties of the scoring regimes established in Section~\ref{sec:student_behavior} and Section~\ref{sec:measurement} remain intact. The core results generalize to convex costs ($\alpha > 1$) as follows:

\paragraph{The Safety Net Effect (Proposition~\ref{prop:safetynet})} This relies on the statistical properties of the $\max$ operator. Pointwise dominance ($V_2^{SC} \ge V_2^{SS}$) holds strictly regardless of cost curvature.

Under Single-Sitting, the continuation objective evaluates the new sitting holistically, while under Superscoring, the continuation objective retains the maximum subject score across attempts. By the fundamental property of the maximum operator, $\max\{A, B\} \ge B$. Therefore, for any arbitrary choice of second-round effort and any realized first-round scores $(X_1, Y_1)$, the continuation value is pointwise weakly greater: $V_2^{SC}(X_1,Y_1) \ge V_2^{SS}(X_1,Y_1)$. Because the retake region expands under Superscoring ($\mathcal{R}^{SS} \subseteq \mathcal{R}^{SC}$), probability mass shifts to states with lower marginal expected returns (the stopping state versus the bounded retaking state). By the strict concavity of the Stage 1 objective, this downward shift in the expected marginal benefit integral guarantees that Round 1 effort is strictly lower under Superscoring, preserving the Safety Net Effect entirely independent of the cost curvature parameter $\alpha$.

\paragraph{Monotonicity in Retake Costs (Proposition~\ref{prop:monotone})} This holds because the fixed retake cost $C(W_i)$ enters only the marginal benefit side of the Stage 1 objective and does not interact with the curvature parameter $\alpha$.

Applying the Envelope Theorem to the Stage 1 expected utility objective isolates the impact of the retake option to the boundary of the Retake Region $\mathcal{R}(C) = \{ \boldsymbol{\delta}_1 : V_2^R > S_1^R \}$. The fixed cost $C$ is paid exclusively upon retaking, entering the continuation value purely as an additive constant. Therefore, it does not alter the marginal integrands; an increase in $C$ solely causes the retake region $\mathcal{R}(C)$ to strictly shrink. Because the marginal return under continuation is strictly lower than the marginal return under stopping ($\frac{\partial V_2^R}{\partial x_1} < \frac{\partial S_1^R}{\partial x_1}$), the expectation integrates a strictly negative quantity over this region. Shrinking this region makes the integral less negative, resulting in a strictly positive cross-partial derivative ($\frac{\partial^2 V_1^R}{\partial x_1 \partial C} > 0$). This guarantees $\frac{\partial x_1^*}{\partial C} > 0$, independent of the specific functional form of the variable costs.

\paragraph{Effort Specialization (Proposition~\ref{prop:specialization})} The corner solution logic extends directly because it depends on the piecewise structure of the Superscoring objective, not on linearity. While the exact threshold becomes a path-dependent target under convexity rather than a fixed constant, the underlying discontinuity remains active, driving the asymmetric Stage 2 effort concentrations observed computationally in Figure~\ref{fig:effort_decomp_empirical}.

Because Superscoring evaluates sections independently, the marginal benefit of studying a banked subject (where $f(\eta_M(x_1+x_2)) < X_1$) drops exactly to zero. Because the marginal cost of strictly positive effort under strictly convex costs is strictly positive ($\beta k(x_2+y_2)^{\alpha-1} > 0$), the net marginal return is strictly negative at the boundary. Therefore, the Kuhn-Tucker conditions dictate a strict corner solution ($x_2^* = 0$) for the banked subject. The student must allocate all positive effort exclusively to the deficient subject in order to offset the fixed cost $\beta C$, perfectly mirroring the specialization mechanics established under linear costs.

\paragraph{The Access Premium (Proposition~\ref{prop:access})} The lower-bound argument only requires that the unconstrained student's feasible set strictly contains the constrained student's set, and that the expected noise boost $\Omega_j > 0$. Both properties hold regardless of cost curvature, resulting in the ability-adjusted score gaps demonstrated in Figure~\ref{fig:access_premium}.

\textit{(The formal extension of this lower-bound argument is provided directly following the Access Premium proof in Appendix~\ref{subsec:propaccess}).}

\paragraph{Partial Generalizations and Divergent Total Effort}
Proposition~\ref{prop:scorebias} partially generalizes. The algorithmic score inflation component ($\mathbb{E}[S^{SC}_{final}] > \mathbb{E}[S^{SS}_{final}]$) relies only on order statistics and holds for any $\alpha \ge 1$. 

However, the total effort reduction component ($\mathbb{E}[e^{SC}] < \mathbb{E}[e^{SS}]$) does not generalize cleanly to strict convexity. Under convex costs, Superscoring breaks the superadditive penalty of Single-Sitting (where studying both subjects simultaneously causes costs to explode). This exhaustion relief mechanism allows constrained students to invest substantially more effort in their deficient dimension. Consequently, the net direction of total effort under convex costs is parameter-dependent, a dynamic explicitly quantified in Figure \ref{fig:effort_decomp_empirical}, which shows Round 2 effort increasing under Superscoring.

\section{Proofs from Section~\ref{sec:measurement}: Consequences for Measurement}
\subsection{Proposition~\ref{prop:screening}: Screening Distortion}
\label{subsec:propscreening}
\begin{proof}
We first compare the expected submitted score under Single-Sitting and Superscoring for a student who takes the exam twice ($N_i = 2$), holding their effort allocation $(x_1, y_1)$ and $(x_2, y_2)$ constant to isolate the mechanical measurement distortion.

Under Single-Sitting (SS), the student submits the higher composite:
$$S_i^{SS} = \max(w_M X_1 + w_V Y_1,\; w_M X_2 + w_V Y_2).$$

Under Superscoring (SC), the student submits the subject-level maxima:
$$S_i^{SC} = w_M \max(X_1, X_2) + w_V \max(Y_1, Y_2).$$

For any realization, by definition of the maximum operator, $\max(X_1,X_2) \ge X_t$ and $\max(Y_1,Y_2) \ge Y_t$ for both $t \in \{1,2\}$. Multiplying by the strictly positive subject weights and summing yields:
\begin{align*}
    S_i^{SC} & = w_M \max(X_1,X_2) + w_V \max(Y_1,Y_2) \\
    & \ge w_M X_t + w_V Y_t \quad \text{for } t \in \{1,2\}.
\end{align*}

Because $S_i^{SC}$ is greater than or equal to both individual composites, it must be greater than or equal to their maximum. Therefore, $S_i^{SC} \ge S_i^{SS}$ pointwise for every realization of testing noise.

Strict inequality holds on the event $\{X_1 > X_2 \text{ and } Y_2 > Y_1\}$ (or vice versa), which occurs when the highest subject scores are drawn from different sittings. In our endogenous effort model, $X_2$ and $X_1$ have different deterministic means if $x_2 > 0$. Specifically, $X_1 > X_2$ occurs if and only if the Round 1 noise advantage overcomes the Round 2 effort accumulation: 
$$\delta_{M1} - \delta_{M2} > f(\eta_M(x_1+x_2)) - f(\eta_M x_1).$$
Because the noise terms $\delta$ are drawn from a continuous, symmetric distribution with non-degenerate variance ($\text{Var}(\delta) > 0$), this event occurs with strictly positive probability. Because the Math and Verbal shocks are mutually independent, the joint event $\Prob(X_1 > X_2 \cap Y_2 > Y_1) > 0$. Taking expectations over this strictly positive probability space yields:
$$\mathbb{E}[S_i^{SC}] > \mathbb{E}[S_i^{SS}].$$

Next, we establish the distortion relative to a single attempt. By the property of the maximum operator, the Single-Sitting composite is pointwise weakly greater than the Round 1 composite: $\max(S_1, S_2) \ge S_1$. Strict inequality occurs whenever $S_2 > S_1$. Again, due to the non-degenerate continuous noise, the probability that the Round 2 realization strictly exceeds Round 1 is strictly positive. Therefore:
$$\mathbb{E}[S_i^{SS}] = \mathbb{E}[\max(S_1, S_2)] > \mathbb{E}[S_1].$$

This proves the strict three-part inequality for $N_i = 2$. The inflation arises mechanically from selection over independent noise across attempts. 
    
The result extends to any $N_i \ge 2$ attempts by mathematical induction. Assume the strict inequalities hold for $N_i=n$ attempts. Adding attempt $n+1$ produces raw scores $(X_{n+1}, Y_{n+1})$. Under Superscoring, the new submitted score is $w_M\max(X_1,\ldots,X_{n+1})+w_V\max(Y_1,\ldots,Y_{n+1}) \ge w_M\max(X_1,\ldots,X_n)+w_V\max(Y_1,\ldots,Y_n)$ pointwise. Strict inequality occurs on the event $\{X_{n+1}>\max(X_1,\ldots,X_n)\} \cup \{Y_{n+1}>\max(Y_1,\ldots,Y_n)\}$, which has strictly positive probability by non-degeneracy. Therefore, the expected submitted score under either rule is strictly increasing in the number of attempts $N_i$.
\end{proof}

\subsection{Proposition~\ref{prop:scorebias}: Signal Degradation under Superscoring}
\label{subsec:propscorebias}
\begin{proof}
We first establish that expected total effort is strictly lower under Superscoring. Let $e^R = (x_1^R + y_1^R) + (x_2^R + y_2^R)$ denote total effort across both rounds under rule $R \in \{SS, SC\}$. 

First-round effort is strictly lower under SC ($x_1^{SC} + y_1^{SC} < x_1^{SS} + y_1^{SS}$) by the Safety Net Effect (Proposition~\ref{prop:safetynet}), which is active whenever the retake probability is strictly interior ($0 < \Prob(\text{Retake}) < 1$).

Conditional on retaking, second-round effort is also strictly lower under SC. Under Single-Sitting, the student targets the interior proficiency levels $E_M^*$ and $E_V^*$ (Lemma~\ref{lem:target}). Because Round 1 effort was strictly lower under SC, the effort required to reach these targets would mechanically be higher. However, under Superscoring, the student may hit a corner solution in a banked subject (Proposition~\ref{prop:specialization}). When this corner solution is active, Round 2 effort in the banked subject drops completely to zero ($x_2^{SC}=0$), whereas under Single-Sitting, the student would be forced to study to replace the lost score. Therefore, across the full probability space, expected total effort is strictly lower under SC:
$$ \mathbb{E}[e^{SC}] < \mathbb{E}[e^{SS}]. $$

Next, we evaluate the expected final submitted score. As established in the proof of Proposition~\ref{prop:screening}, for any fixed, arbitrary effort allocation, the Superscoring composite pointwise dominates the Single-Sitting composite due to the convexity of the subject-level maximum operator. We term this the \textit{mechanical score boost}.

However, because effort is endogenous, $e^{SC} < e^{SS}$, leading to an effort reduction penalty. The expected submitted score under SC strictly exceeds SS if and only if the mechanical score boost outweighs the effort reduction penalty.

Because $f(\cdot)$ is strictly concave, the marginal loss of score due to effort reduction is bounded. Specifically, the maximum possible score loss is bounded above by the total deterministic score under Single-Sitting, $w_M f(E_M^*) + w_V f(E_V^*)$, which is finite since $f(0)=0$ and optimal effort targets are finite (Lemma~\ref{lem:target}). 

Conversely, the mechanical score boost of Superscoring scales directly with the variance of the testing noise $\text{Var}(\delta)$. By taking the maximum of independent subject-level draws across two sittings, Superscoring systematically harvests positive shocks. For any fixed knowledge level, the expected value of $\max(X_1, X_2)$ is strictly increasing in the dispersion of $\delta$. 

Therefore, provided the noise variance $\text{Var}(\delta)$ is sufficiently large relative to the cost of effort $k$ (which bounds the effort reduction penalty), the expected mechanical boost strictly dominates the score loss from reduced effort:
$$ \mathbb{E}[S^{SC}] > \mathbb{E}[S^{SS}]. $$

When this parameter condition holds, Superscoring exhibits strict signal degradation: universities observe strictly higher submitted scores $\mathbb{E}[S^{SC}] > \mathbb{E}[S^{SS}]$, yet those scores are generated by strictly lower total human capital investment $\mathbb{E}[e^{SC}] < \mathbb{E}[e^{SS}]$.
\end{proof}

\subsection{Corollary~\ref{cor:sbr}: Score Bias Ratio}
\label{app:corsbr}
While standard signal processing evaluates the variance of noise, universities setting strict admission cutoffs primarily care about systematic upward bias—whether an expected score systematically overstates true ability. To quantify this refinement of Proposition~\ref{prop:scorebias}, we define the Score Bias Ratio (SBR) as true composite ability relative to the expected submitted score: 
\begin{align*} 
    \mathrm{SBR}_i^R \equiv \frac{\eta_i}{\mathbb{E}[S_i]}, 
\end{align*} 
where $\eta_i$ is calibrated so that $\mathbb{E}[S_{i}]=\eta_i$ for a single-attempt student. Under this calibration, $\mathrm{SBR}_i^{SS}=1$ when $N_i=1$.

\begin{corollary}[Score Bias Ratio]
\label{cor:sbr} 
For any student taking the exam multiple times $(N_i\ge 2)$: 
\begin{align*} 
    \mathrm{SBR}_i^{SC} < \mathrm{SBR}_i^{SS} < 1.
\end{align*} 
Both ratios are strictly decreasing in the number of attempts $N_i$. 
\end{corollary}

\begin{proof}
By the calibration assumption, a single-attempt student ($N_i=1$) submits an expected score equal to their true ability: $\mathbb{E}[S_1] = \eta_i$, yielding $\mathrm{SBR}_i^R = 1$.

Proposition~\ref{prop:screening} (Screening Distortion) establishes that for any student taking the exam $N_i \ge 2$ times, the expected submitted score strictly exceeds the single-attempt expected score, and the Superscoring composite strictly exceeds the Single-Sitting composite:
$$ \mathbb{E}[S_i^{SC}] > \mathbb{E}[S_i^{SS}] > \mathbb{E}[S_1] = \eta_i. $$

Dividing the numerator $\eta_i$ by these strictly ordered expectations directly yields the bias inequalities:
$$ \mathrm{SBR}_i^{SC} < \mathrm{SBR}_i^{SS} < 1. $$

Furthermore, Proposition~\ref{prop:screening} proves by induction that the expected submitted score strictly increases with the number of attempts $N_i$. Because the denominator strictly increases with $N_i$, the ratio $\mathrm{SBR}_i^R$ is strictly decreasing in $N_i$ under both rules.

We now evaluate the disparate impact across wealth $\mathcal{W}_i$. In our model, the fixed cost of retaking $C(\mathcal{W}_i)$ is strictly decreasing in wealth. As established in the proof of Proposition~\ref{prop:monotone}, a decrease in fixed cost $C$ strictly expands the retake region $\mathcal{R}$ in the noise space. Therefore, wealthier students have a strictly higher equilibrium probability of retaking the exam, leading to a higher expected number of attempts. 

Because the Score Bias Ratio is strictly decreasing in retake probability, the bias is disparate across wealth: wealthier students systematically achieve lower $\mathrm{SBR}$ values (i.e., higher artificial score inflation relative to true ability). 

Finally, because the mechanical score boost of the subject-level maximum operator is strictly larger than the composite-level operator (Proposition~\ref{prop:screening}), the inflation harvested per retake is strictly greater under Superscoring. Therefore, the wealth-driven bias divergence is strictly magnified under Superscoring, causing the greatest relative harm to lower-wealth populations who are constrained by cost to a single attempt.
\end{proof}

\subsection{Proposition~\ref{prop:access}: Access Premium}
\label{subsec:propaccess}
\begin{proof}
Let Student B be the constrained student ($C^B \to \infty$). Restricted to a single attempt ($N^B=1$), they optimally choose effort $(x_1^B, y_1^B)$ and Stage 2 effort $(0,0)$. Because testing noise is mean-zero ($\mathbb{E}[\delta_j] = 0$), their expected final score is deterministic:
$$ \mathbb{E}[S^B] = w_M f(\eta_M x_1^B) + w_V f(\eta_V y_1^B). $$

Let Student A be the unconstrained student ($C^A < \infty$) with identical ability $(\eta_M, \eta_V)$. Because Student A faces a lower fixed cost, their feasible choice set strictly contains the single-attempt strategy of Student B.

Consider a lower-bound "mimicking" strategy: Student A chooses the exact Round 1 effort as Student B $(x_1^B, y_1^B)$, but exerts zero effort in Round 2, relying purely on Superscoring to harvest the maximum of the two independent noise draws. The expected score of this feasible strategy strictly exceeds Student B's score due to the convexity of the maximum operator:
\begin{align*}
     \mathbb{E}[S^{LB}] & = w_M \mathbb{E}[\max(X_{M1}, X_{M2})]
     +  w_V \mathbb{E}[\max(X_{V1}, X_{V2})] 
     > \mathbb{E}[S^B]. 
\end{align*}
Let the difference be the pure noise boost: $\Omega = \mathbb{E}[S^{LB}] - \mathbb{E}[S^B] > 0$.

Because Student A optimally chooses effort across both rounds to maximize expected utility (Score - Costs), they will only deviate from this mimicking strategy if the deviation yields strictly higher expected utility.

If Student A's optimal expected score was weakly lower than Student B's ($\mathbb{E}[S^A] \le \mathbb{E}[S^B]$), they would be sacrificing at least $\Omega$ in expected score compared to the mimicking strategy. Because the cost function is strictly increasing, such a massive score sacrifice is only rational if it generates immense effort cost savings. However, as established in Proposition~\ref{prop:scorebias}, provided the variance of the testing noise $\text{Var}(\delta)$ is sufficiently large relative to the marginal cost of effort, the mechanical score boost of taking the maximum over multiple draws dominates any rational effort reduction penalty.

Therefore, under sufficiently dispersed noise, Student A's utility-maximizing strategy must yield a strictly higher expected score:
$$ \mathbb{E}[S^A] > \mathbb{E}[S^B]. $$
This confirms a strictly positive Access Premium.
\end{proof}

\subsection{The Compensatory Ability Gap}
\label{app:abilitygap}
Building upon the Access Premium established in Proposition~\ref{prop:access}, we formally quantify the disparity in latent ability required for a constrained student to mathematically match the expected score of an unconstrained student.

\begin{corollary}[The Ability Gap]
\label{cor:abilitygap}
To close the Access Premium, the constrained Student B must possess strictly higher latent ability $\tilde{\eta}^B > \eta^A$. Assuming identical baseline Verbal ability, the required compensatory Mathematics ability is:
\begin{align*}
    \tilde{\eta}_M^B = \frac{1}{x_1^B} f^{-1}\Bigl( f(\eta_M^A x_1^B) + \Omega_M \Bigr),
\end{align*}
where $\Omega_M$ is the effective unweighted Math score boost harvested by Student A. A symmetric condition holds for Verbal.
\end{corollary}

\begin{proof}
By Proposition~\ref{prop:access}, the unconstrained Student A achieves a strictly higher expected score $\mathbb{E}[S^A]$ than the constrained Student B ($\mathbb{E}[S^B]$) when both possess identical latent ability $\eta^A = \eta^B$. The difference is the Access Premium: $\Delta S = \mathbb{E}[S^A] - \mathbb{E}[S^B] > 0$.

Because Student B is restricted to a single attempt, their expected score is purely deterministic based on their Round 1 effort: 
$$\mathbb{E}[S^B] = w_M f(\eta_M^A x_1^B) + w_V f(\eta_V^A y_1^B).$$

For Student B to perfectly match Student A's expected submitted score ($\mathbb{E}[S^B_{new}] = \mathbb{E}[S^A]$) using only a single attempt, they must possess a compensatory latent ability $\tilde{\eta}_M^B$. Assuming the compensation occurs entirely in Mathematics, we set their new expected score equal to Student A's expected score:
$$ w_M f(\tilde{\eta}_M^B x_1^B) + w_V f(\eta_V^A y_1^B) = \mathbb{E}[S^B] + \Delta S $$
$$ w_M f(\tilde{\eta}_M^B x_1^B) = w_M f(\eta_M^A x_1^B) + \Delta S $$

Dividing by the subject weight $w_M$ isolates the unweighted Math points. We define this term as $\Omega_M \equiv \frac{\Delta S}{w_M} > 0$, representing the effective score boost harvested by Student A, normalized to the Mathematics scale. Substituting $\Omega_M$ yields:
$$ f(\tilde{\eta}_M^B x_1^B) = f(\eta_M^A x_1^B) + \Omega_M $$

Applying the inverse knowledge production function $f^{-1}(\cdot)$ (which exists and is strictly increasing because $f'>0$) yields:
$$ \tilde{\eta}_M^B x_1^B = f^{-1}\Bigl( f(\eta_M^A x_1^B) + \Omega_M \Bigr) $$

Dividing by the effort $x_1^B$ isolates the required ability:
$$ \tilde{\eta}_M^B = \frac{1}{x_1^B} f^{-1}\Bigl( f(\eta_M^A x_1^B) + \Omega_M \Bigr). $$

Because $f(\cdot)$ is strictly increasing, its inverse $f^{-1}(\cdot)$ is also strictly increasing. Since the effective boost is strictly positive ($\Omega_M > 0$), the argument inside the inverse function is strictly larger than $f(\eta_M^A x_1^B)$. 

Therefore, it mathematically must hold that $\tilde{\eta}_M^B > \eta_M^A$. Under Superscoring, universities enforcing strict score cutoffs systematically screen out higher-ability constrained applicants in favor of lower-ability unconstrained applicants.
\end{proof}

\subsection{Proposition~\ref{prop:accelerant}: Superscoring as an Inequality Accelerant}
\label{subsec:propaccelerant}
\begin{proof}
To isolate the algorithmic effect of the scoring rules from underlying human capital, we evaluate two demographic groups $H$ and $L$ with identical latent ability. Because the fixed cost of retaking $C(W_i)$ is strictly decreasing in wealth, the High Wealth group faces a strictly lower cost ($C_H < C_L$). By the Monotonicity condition (Proposition~\ref{prop:monotone}), the High Wealth group optimally chooses a strictly higher attempt frequency: $N_H > N_L \ge 1$.

Let $\Delta(N) \equiv \mathbb{E}[S^{SC}(N)] - \mathbb{E}[S^{SS}(N)]$ denote the expected algorithmic inflation—the difference in expected submitted score between Superscoring and Single-Sitting—for a student who takes the exam $N$ times. 

We want to prove that the submitted-score gap under SC is strictly greater than the gap under SS:
$$ \mathbb{E}[S^{SC}(N_H)] - \mathbb{E}[S^{SC}(N_L)] > \mathbb{E}[S^{SS}(N_H)] - \mathbb{E}[S^{SS}(N_L)]. $$
Rearranging the terms, this inequality is mathematically equivalent to proving that the algorithmic inflation strictly expands with the number of attempts:
\begin{align*}
    & \mathbb{E}[S^{SC}(N_H)] - \mathbb{E}[S^{SS}(N_H)] > \mathbb{E}[S^{SC}(N_L)] - \mathbb{E}[S^{SS}(N_L)], \\
    & \text{or} \quad \Delta(N_H) > \Delta(N_L).
\end{align*}

Because $N_H > N_L$, this holds if and only if $\Delta(N)$ is strictly increasing in $N$. We proceed by showing that the expected marginal gain of an additional attempt $(N+1)$ is strictly larger under Superscoring than Single-Sitting. 

Let $X_{MAX}$ and $Y_{MAX}$ be the highest individual subject scores achieved over the first $N$ attempts. Let $S_{MAX}$ be the highest Single-Sitting composite achieved over the first $N$ attempts. From Proposition~\ref{prop:screening}, we established pointwise dominance for any $N$: 
$$S_{MAX} \le w_M X_{MAX} + w_V Y_{MAX}.$$

Now consider the $(N+1)$-th attempt, yielding new scores $(X_{new}, Y_{new})$. 
Under Superscoring, the marginal score gain isolates improvement in either subject independently:
\begin{align*}
    \text{Gain}^{SC} & = w_M \max(0, X_{new} - X_{MAX}) 
    + w_V \max(0, Y_{new} - Y_{MAX}).
\end{align*}

Under Single-Sitting, the marginal score gain requires the new composite to beat the old composite maximum:
$$ \text{Gain}^{SS} = \max\big(0, (w_M X_{new} + w_V Y_{new}) - S_{MAX}\big). $$

Because $S_{MAX} \le w_M X_{MAX} + w_V Y_{MAX}$, replacing $S_{MAX}$ with the larger Superscoring components weakly decreases the term inside the maximum:
\begin{align*}
     \text{Gain}^{SS} & \le \max\big(0, (w_M X_{new} + w_V Y_{new})
     - (w_M X_{MAX} + w_V Y_{MAX})\big). 
\end{align*}
Rearranging the terms algebraically:
\begin{align*}
    \text{Gain}^{SS} & \le \max\big(0, w_M (X_{new} - X_{MAX})
    + w_V (Y_{new} - Y_{MAX})\big).
\end{align*}

By the subadditivity of the maximum operator, $\max(0, A + B) \le \max(0, A) + \max(0, B)$. Applying this to the weighted score differences yields:
\begin{align*}
    \text{Gain}^{SS} &  \le w_M \max(0, X_{new} - X_{MAX}) 
    + w_V \max(0, Y_{new} - Y_{MAX}) 
    = \text{Gain}^{SC}.
\end{align*}

Therefore, the marginal gain of every additional attempt is pointwise weakly greater under Superscoring. 

Strict inequality holds in expectation because the testing shocks $\delta_M$ and $\delta_V$ are independent and continuous. There is a strictly positive probability that the student receives a highly positive shock in Mathematics but a highly negative shock in Verbal (i.e., $X_{new} > X_{MAX}$ but $Y_{new} \ll Y_{MAX}$). In this event, the Single-Sitting gain is exactly $0$ because the bad Verbal score drags the composite down. However, the Superscoring gain is strictly positive because it banks the Math score and discards the Verbal score.

Taking expectations over this strictly positive probability space yields:
$$ \mathbb{E}[\text{Gain}^{SC}] > \mathbb{E}[\text{Gain}^{SS}]. $$

Because the expected marginal gain of every additional attempt is strictly larger under Superscoring, the algorithmic inflation $\Delta(N)$ is strictly increasing in $N$. Since $N_H > N_L$, the inflation harvested by the High Wealth group strictly exceeds the inflation harvested by the Low Wealth group ($\Delta(N_H) > \Delta(N_L)$), acting as a mathematical accelerant on the pre-existing wealth inequality.
\end{proof}

\subsection{Proposition~\ref{prop:participation}: Participation Effect}
\label{subsec:propparticipation}
\begin{proof}
We compare the maximum expected utility of participating under Single-Sitting ($V^{SS*}$) and Superscoring ($V^{SC*}$), assuming the outside option of exiting the pool yields a utility of $0$. A student participates under regime $\mathcal{R} \in \{SS,SC\}$ if and only if:
$$V^{\mathcal{R}*} \ge 0.$$
To isolate the pure convexity mechanism of the multidimensional effort penalty, we evaluate a marginal student for whom the fixed cost of testing is negligible ($C \to 0$). All participation differences therefore arise solely from the effort cost function $c(x,y) = \frac{k}{\alpha}(x+y)^\alpha$.

Consider an arbitrary target effort allocation $(\bar{x}, \bar{y})$ where $\bar{x} > 0$ and $\bar{y} > 0$. Let $S_{target}$ denote the deterministic baseline score generated by this effort: 
$$S_{target} = w_M f(\eta_M \bar{x}) + w_V f(\eta_V \bar{y}).$$

Under Single-Sitting (SS), the student must deliver this effort in a single round. The total expected utility of this strategy is the expected score minus the simultaneous effort cost:
$$V^{SS}(\bar{x}, \bar{y}) = \mathbb{E}[S^{SS}] - \frac{k}{\alpha}(\bar{x} + \bar{y})^\alpha.$$
By definition, $\mathbb{E}[S^{SS}] \ge S_{target}$.

Under Superscoring (SC), the student can allocate this same effort across two rounds using a "split strategy":
\begin{enumerate}
    \item Round 1: Exert effort $(x_1,y_1) = (\bar{x},0)$.
    \item Round 2: Exert effort $(x_2,y_2) = (0,\bar{y})$.
\end{enumerate}

Under this strategy, the Round 1 Verbal score is $Y_1 = f(\eta_V \cdot 0) + \delta_{V1} = \delta_{V1}$, which is pure mean-zero noise. The student intentionally forgoes Round 1 Verbal preparation, relying on Superscoring to bank the Round 2 Verbal score $Y_2 = f(\eta_V \bar{y}) + \delta_{V2}$. The expected submitted Verbal score is $\mathbb{E}[\max(\delta_{V1}, f(\eta_V \bar{y}) + \delta_{V2})] \ge f(\eta_V \bar{y})$. Because this logic is symmetric for Mathematics, the expected final submitted score under this split strategy strictly bounds $S_{target}$ from below: $\mathbb{E}[S^{SC}_{split}] \ge S_{target}$.

The total cost of this split strategy under Superscoring is:
$$C^{SC}_{split} = \frac{k}{\alpha}\bar{x}^\alpha + \beta \frac{k}{\alpha}\bar{y}^\alpha.$$

Because the effort cost function is strictly convex ($\alpha > 1$), the function $z^\alpha$ is strictly superadditive:
$$(\bar{x} + \bar{y})^\alpha > \bar{x}^\alpha + \bar{y}^\alpha \quad \text{for all } \bar{x},\bar{y} > 0.$$

Since the discount factor $\beta \le 1$, the cost of simultaneous preparation strictly exceeds the cost of sequential preparation:
$$ \frac{k}{\alpha}(\bar{x} + \bar{y})^\alpha > \frac{k}{\alpha}\bar{x}^\alpha + \beta \frac{k}{\alpha}\bar{y}^\alpha \implies C^{SS} > C^{SC}_{split}. $$

Because the split strategy generates a weakly higher expected score at a strictly lower cost, its expected utility strictly dominates the Single-Sitting strategy:
$$V^{SC}_{split} > V^{SS}(\bar{x}, \bar{y}).$$

Let $V^{SC*}$ and $V^{SS*}$ denote the globally maximized expected utilities under each rule. Because $V^{SC*} \ge V^{SC}_{split}$, it follows that $V^{SC*} > V^{SS*}$. 

By continuity, there exists a non-empty parameter space of cost multipliers ($k$) or base abilities ($\eta$) for which the utility of Single-Sitting is strictly negative, but the utility of Superscoring is strictly positive:
$$V^{SC*} > 0 > V^{SS*}.$$
For these constrained students, the simultaneous effort penalty of Single-Sitting forces them to exit the pool ($A_i(SS)=0$). Superscoring breaks this superadditive penalty, inducing participation ($A_i(SC)=1$).

Finally, under linear costs ($\alpha=1$), $(x+y)^1 = x^1 + y^1$. The strict cost advantage of splitting disappears. Because any Round 2 attempt would strictly incur the fixed cost $\beta C > 0$ without providing any physical effort discount, sequential splitting is strictly cost-dominated, meaning Superscoring cannot induce participation through this mechanism when $\alpha=1$.
\end{proof}

\subsection{Proposition~\ref{prop:tradeoff}: Institutional Trade-Off}
\label{subsec:proptradeoff}
\begin{proof}
Let $V_{uni}(R)$ denote the expected ability of the top $K$ admitted students under regime $R \in \{SS, SC\}$.

Under SS, the simultaneous effort penalty forces all Type A students to exit the pool ($A_i(SS) = 0$). The university admits the top $K$ students from the Type B pool, yielding an expected ability $V_{uni}(SS) = \bar{\eta}_B$.

Under SC, both types participate. Type A students submit a single-attempt score $S_A = \eta_A + \delta_1$. Unconstrained Type B students optimize by taking the exam twice, submitting $S_B = \eta_B + \max(\delta_1, \delta_2)$.

We evaluate the non-trivial case where $\mathbb{E}[S_A] > \mathbb{E}[S_B]$; if $\mathbb{E}[S_B] > \mathbb{E}[S_A]$, Type B noise harvesters saturate the upper tail, and SC is strictly dominated for all $\gamma$. Because Type A retains an expected score advantage, the top $K$ seats are filled by a mixture of both types. Let $P_A(\gamma) \in [0,1]$ be the expected proportion of the seats captured by Type A students, and $s^*(\gamma)$ be the endogenous admission cutoff satisfying the capacity constraint.

The university's expected payoff is:
$$V_{uni}(SC, \gamma) = P_A(\gamma)\eta_{A} + (1 - P_A(\gamma))\eta_{noisy}(\gamma),$$
where $\eta_{noisy}(\gamma) \equiv \mathbb{E}[\eta_B \mid S_B \ge s^*(\gamma)]$.

Because $\max(\delta_1, \delta_2)$ first-order stochastically dominates $\delta_1$, conditional on clearing a high threshold, a larger fraction of a Superscored composite is attributable to noise rather than true ability. The signal fidelity is strictly degraded, yielding $\eta_{noisy}(\gamma) < \bar{\eta}_B$.

We establish strict monotonicity with respect to $\gamma$. The total derivative is:
\[
    \frac{d V_{uni}}{d\gamma} = \frac{d P_A}{d \gamma} \big(\eta_A - \eta_{noisy}(\gamma)\big) + \big(1-P_A(\gamma)\big) \frac{d \eta_{noisy}}{d \gamma}.
\]
Adding Type A mass strictly increases the expected proportion of the top-$K$ order statistics drawn from the Type A distribution, so $\frac{d P_A}{d \gamma} > 0$. Since $\eta_A > \bar{\eta}_B > \eta_{noisy}(\gamma)$, the first term is strictly positive.
For the second term, as $\gamma$ increases, the cutoff $s^*(\gamma)$ tightens. Assuming standard log-concave testing noise (satisfying the Monotone Likelihood Ratio Property), a strictly higher score cutoff screens more stringently on the true ability $\eta_B$ within the noisy distribution, ensuring $\frac{d \eta_{noisy}}{d \gamma} \ge 0$. Therefore, both terms are non-negative, and the first is strictly positive, establishing $\frac{d V_{uni}}{d \gamma} > 0$.

We evaluate the limits:
\begin{enumerate}
        \item As $\gamma \to 0$, $P_A \to 0$, yielding $V_{uni}(SC) \to \eta_{noisy} < V_{uni}(SS)$.
        \item As $\gamma \to 1$ (assuming capacity $K$ is small relative to the total population), $P_A \to 1$. $V_{uni}(SC) \to \eta_{A} > V_{uni}(SS)$.
\end{enumerate}

Because $V_{uni}(SC,\gamma)$ is continuous, strictly monotonic, and bounded across $V_{uni}(SS)$, the Intermediate Value Theorem guarantees a unique threshold $\gamma^* \in (0,1)$ such that $V_{uni}(SC, \gamma^*) = V_{uni}(SS)$. It follows that $V_{uni}(SC) > V_{uni}(SS)$ if and only if $\gamma > \gamma^*$.
\end{proof}

\section{Properties of Alternative Mechanisms}
\label{app:mechanisms}
This section provides the formal theoretical characterization of the alternative scoring mechanisms introduced in Section~\ref{subsec:alternative_rules}. 

Consider an additive admissions correction: $S_i^{R} = S_i^{SC} - g(\mathcal{I}_i)$, where $\mathcal{I}_i$ is the information the university observes and $g$ is any measurable penalty function. The institution designs $g^*$ to minimize the mean squared error (MSE) between the corrected score and latent ability $\eta_i$. Let $\Delta_i \equiv S_i^{SC} - \eta_i$ denote the uncorrected mechanical distortion.

\begin{theorem}[Optimal Information-Constrained Correction]
\label{thm:optimal}
Within the additive correction class, the MSE-minimizing correction is:
\begin{align*} 
    g^*(\mathcal{I}_i)=\mathbb{E}[\Delta_i \mid \mathcal{I}_i]. 
\end{align*}

This estimator is conditionally unbiased: $\mathbb{E}[S_i^R \mid \eta_i, \mathcal{I}_i] = \eta_i$. For any information refinement $\mathcal{I}_i \subseteq \mathcal{I}_i'$:
\begin{align*}
    \mathrm{Var}(S_i^R(\mathcal{I}_i')-\eta_i)\le\mathrm{Var}(S_i^R(\mathcal{I}_i)-\eta_i).
\end{align*}
\end{theorem}

\begin{proof}
Let $S_i^R(\mathcal{I}_i) = S_i^{SC} - g(\mathcal{I}_i)$ denote the corrected score. The institution seeks a measurable function $g$ that minimizes the mean squared error conditional on the observed information set $\mathcal{I}_i$:
$$ \min_g \mathbb{E}\left[ \big(S_i^{SC} - g(\mathcal{I}_i) - \eta_i\big)^2 \mid \mathcal{I}_i \right]. $$

By definition, the uncorrected distortion is $\Delta_i \equiv S_i^{SC} - \eta_i$. Substituting this into the objective yields:
$$ \min_g \mathbb{E}\left[ \big(\Delta_i - g(\mathcal{I}_i)\big)^2 \mid \mathcal{I}_i \right]. $$

Under squared-error loss, it is a standard statistical property that the conditional expectation is the unique global minimizer. Therefore, the optimal additive correction is:
$$ g^*(\mathcal{I}_i) = \mathbb{E}[\Delta_i \mid \mathcal{I}_i]. $$

Applying this optimal correction, the residual error is $(S_i^R - \eta_i) = \Delta_i - \mathbb{E}[\Delta_i \mid \mathcal{I}_i]$. Taking the conditional expectation verifies that the estimator is conditionally unbiased:
\begin{align*}
    \mathbb{E}[S_i^R - \eta_i \mid \mathcal{I}_i] 
    = \mathbb{E}\big[\Delta_i - \mathbb{E}[\Delta_i \mid \mathcal{I}_i] \mid \mathcal{I}_i\big]
    = \mathbb{E}[\Delta_i \mid \mathcal{I}_i] - \mathbb{E}[\Delta_i \mid \mathcal{I}_i]
    = 0.
\end{align*}
Because $\mathbb{E}[S_i^R - \eta_i \mid \eta_i, \mathcal{I}_i] = 0$, it follows that $\mathbb{E}[S_i^R \mid \eta_i, \mathcal{I}_i] = \eta_i$.

We now prove the variance ordering. Because the corrected score $S_i^R$ is conditionally unbiased, its conditional variance is equal to the conditional variance of the residual error:
$$ \mathrm{Var}(S_i^R(\mathcal{I}_i) - \eta_i \mid \mathcal{I}_i) = \mathrm{Var}(\Delta_i \mid \mathcal{I}_i). $$
Taking the unconditional expectation of both sides, the total residual variance under the optimal correction $g^*(\mathcal{I}_i)$ is exactly $\mathbb{E}[\mathrm{Var}(\Delta_i \mid \mathcal{I}_i)]$.

Consider a refined information set $\mathcal{I}_i'$ such that $\mathcal{I}_i \subseteq \mathcal{I}_i'$ (i.e., $\mathcal{I}_i'$ contains strictly more information). By the Law of Total Variance applied to the uncorrected distortion $\Delta_i$, conditioning first on the finer information set $\mathcal{I}_i'$: 
$$ \mathrm{Var}(\Delta_i \mid \mathcal{I}_i) = \mathrm{Var}\big(\mathbb{E}[\Delta_i \mid \mathcal{I}_i'] \mid \mathcal{I}_i\big) + \mathbb{E}\big[\mathrm{Var}(\Delta_i \mid \mathcal{I}_i') \mid \mathcal{I}_i\big]. $$

Taking the unconditional expectation of both sides yields: 
\begin{align*}
     \mathbb{E}[\mathrm{Var}(\Delta_i \mid \mathcal{I}_i)] & = \mathbb{E}\Big[\mathrm{Var}\big(\mathbb{E}[\Delta_i \mid \mathcal{I}_i'] \mid \mathcal{I}_i\big)\Big] \\
     & + \mathbb{E}[\mathrm{Var}(\Delta_i \mid \mathcal{I}_i')].
\end{align*}

Because variance is a strictly non-negative operator, the term $\mathbb{E}\big[\mathrm{Var}(\mathbb{E}[\Delta_i \mid \mathcal{I}_i'] \mid \mathcal{I}_i)\big] \ge 0$. Substituting the unconditional residual variances back into the equation, we obtain: 
$$ \mathrm{Var}\big(S_i^R(\mathcal{I}_i) - \eta_i\big) \ge \mathrm{Var}\big(S_i^R(\mathcal{I}_i') - \eta_i\big). $$
This confirms that the residual variance (and therefore the Mean Squared Error) is weakly smaller under the finer information set.
\end{proof}

\subsection{FISC: Full-Information Score Correction}
\label{app:fisc}
The FISC mechanism applies the optimal information-constrained penalty assuming the university perfectly observes the true attempt count $N_i$. Defining the expected mechanical noise premium for $N_i$ attempts as $\Omega(N_i) \equiv \mathbb{E}[\max_{k \le N_i} \delta_k]$, the FISC corrected score is: $S_i^{FISC} = S_i^{SC} - \Omega(N_i)$.

\begin{proposition}[Full-Information Optimality]
\label{prop:fisc}
Assume the university observes all attempts $N_i$ and applies the penalty $\Omega(N_i)$. FISC weakly dominates Single-Sitting because it preserves the participation of constrained high-ability students who would otherwise exit the pool. Furthermore, FISC strictly improves expected admitted ability relative to uncorrected Superscoring by neutralizing attempt-based inflation:
\begin{align*}
V_{uni}(FISC) > V_{uni}(SC) \quad \text{and} \quad V_{uni}(FISC) \ge V_{uni}(SS).
\end{align*}
\end{proposition}

\begin{proof}
We evaluate the university's expected payoff $V_{uni}(\mathcal{R}) = \mathbb{E}[\eta_i \mid i \text{ is admitted}]$ under the Full-Information Score Correction regime (FISC) and compare it against Single-Sitting (SS) and traditional Superscoring (SC). We maintain the applicant pool defined in Proposition~\ref{prop:tradeoff}: high-ability, constrained Type A students ($\eta_A$) and lower-ability, unconstrained Type B students ($\eta_B < \eta_A$).

\paragraph{Weak Dominance over Single-Sitting (SS):}
Under SS, the simultaneous effort penalty forces constrained Type A students to exit the pool ($A_i(SS) = 0$). Under FISC, the underlying aggregation rule before the penalty remains the section-level maximum operator. Therefore, Type A students can still utilize the "split strategy" (established in Proposition~\ref{prop:participation}) to distribute their physical effort sequentially, breaking the superadditive convex effort cost. 

The FISC penalty $\Omega(N_i)$ is a deterministic score translation; it does not reintroduce the physical simultaneous effort constraint. Because the physical cost mechanics are identical to SC, Type A students participate under FISC exactly as they do under SC. 
Because the FISC applicant pool strictly contains the SS applicant pool (by adding high-ability Type A students), and the FISC estimator is conditionally unbiased (Theorem~\ref{thm:optimal}), screening this superior applicant pool yields an expected admitted ability weakly greater than SS: $V_{uni}(FISC) \ge V_{uni}(SS)$.

\paragraph{Strict Dominance over Superscoring (SC):}
Under uncorrected SC, expected scores are artificially inflated by the noise premium. Constrained Type A students split their effort across $N_A=2$ rounds, yielding $\mathbb{E}[S_A^{SC}] = \eta_A + \Omega(2)$. Unconstrained Type B students maximize the algorithmic advantage by taking $N_B > 2$ attempts, yielding $\mathbb{E}[S_B^{SC}] = \eta_B + \Omega(N_B)$. 

This introduces a severe screening distortion: a Type B applicant can displace a Type A applicant with higher true ability in the top $K$ rankings if noise inflation overcomes the ability gap. Let $\mathcal{D}$ denote the positive-measure set of score realizations under SC where a Type B student displaces a Type A student purely due to this asymmetric noise inflation. Because $N_B > N_A \implies \Omega(N_B) > \Omega(N_A)$, the probability of this displacement event is strictly positive under continuous, non-degenerate noise.

Under FISC, the adjustment perfectly internalizes the attempt count and deducts $\Omega(N_i)$. The expected submitted scores become exactly the true abilities:
$$\mathbb{E}[S_A^{FISC}] = \eta_A \quad \text{and} \quad \mathbb{E}[S_B^{FISC}] = \eta_B.$$
FISC restores the conditionally unbiased signal for all applicants, neutralizing the asymmetric expected algorithmic advantage. 

On the set of realizations $\mathcal{D}$, the displacement of Type A by Type B is prevented by the FISC correction. Because the displacing students had strictly lower true ability ($\eta_B < \eta_A$) and are replaced in the admitted set by students with higher true ability, every prevented displacement strictly increases the true latent ability of the admitted class. Integrating over the positive-measure set $\mathcal{D}$ where these swaps occur:
$$V_{uni}(FISC) - V_{uni}(SC) = \frac{1}{K}\int_{\mathcal{D}} (\eta_A - \eta_B) \, d\mathbb{P} > 0.$$
Because $\eta_A > \eta_B$ and $\mathcal{D}$ has strictly positive measure, the integral is strictly positive. Therefore, FISC strictly improves the university's objective relative to uncorrected Superscoring.
\end{proof}

\subsection{Variance Tax}
\label{app:vartax}
For a student with two attempts, the Variance Tax aggregates the scores via a convex combination parameterized by $\lambda \in [0.5,1]$:
\begin{align*}
\widehat S_{VT}(\lambda) = \lambda\max(S_1,S_2)+(1-\lambda)\min(S_1,S_2).
\end{align*}

\begin{proposition}[Variance Tax Score Distribution]
\label{prop:vartax}
Assume a student takes the exam $N_i=2$ times to harvest noise, holding deterministic knowledge $\eta_i$ constant. Under the Variance Tax with parameter $\lambda$, the expected corrected score equals:
\begin{align*}
    \mathbb{E}[\widehat S_{VT}(\lambda)] = \eta_i + (2\lambda-1)\Omega(N_i).
\end{align*}
where $\Omega(2)$ is the binary expected noise premium. For single-attempt students $(N_i=1)$, $\Omega(1)=0$ and the formula yields $\eta_i$ exactly. Setting $\lambda=1$ recovers standard Superscoring; setting $\lambda=0.5$ reduces the correction to the strict attempt average.
\end{proposition}

\begin{proof}
To isolate the algorithmic effect of the Variance Tax on speculative retaking, we evaluate a "pure noise harvester": a student who takes the exam twice ($N_i = 2$) without exerting additional effort in Round 2 ($x_2=0, y_2=0$).

Because their deterministic effort-based knowledge remains constant, their scores in both rounds are centered on the identical baseline ability $\eta_i$:
$$ S_1 = \eta_i + \delta_1 \quad \text{and} \quad S_2 = \eta_i + \delta_2, $$
where the testing shocks $\delta_1, \delta_2$ are drawn i.i.d. from a continuous distribution that is symmetric around zero ($\mathbb{E}[\delta] = 0$).

By the definition of the mechanical noise premium $\Omega(2) \equiv \mathbb{E}[\max(\delta_1, \delta_2)]$, the expectation of the maximum score is:
$$ \mathbb{E}[\max(S_1, S_2)] = \eta_i + \mathbb{E}[\max(\delta_1, \delta_2)] = \eta_i + \Omega(2). $$

Because the underlying noise distribution is symmetric and centered at zero, the expected minimum of the two draws is exactly the negative of the expected maximum:
$$ \mathbb{E}[\min(\delta_1, \delta_2)] = -\mathbb{E}[\max(\delta_1, \delta_2)] = -\Omega(2). $$
Therefore, the expectation of the minimum score is:
$$ \mathbb{E}[\min(S_1, S_2)] = \eta_i - \Omega(2). $$

The university evaluates the student using the convex combination:
$$ \widehat S_{VT}(\lambda) = \lambda \max(S_1, S_2) + (1-\lambda) \min(S_1, S_2). $$

Taking the expectation and applying the linearity of the expectation operator yields:
$$ \mathbb{E}[\widehat S_{VT}(\lambda)] = \lambda\big(\eta_i + \Omega(2)\big) + (1-\lambda)\big(\eta_i - \Omega(2)\big). $$

Expanding and grouping terms by $\eta_i$ and $\Omega(2)$:
$$ \mathbb{E}[\widehat S_{VT}(\lambda)] = \lambda\eta_i + \lambda\Omega(2) + \eta_i - \Omega(2) - \lambda\eta_i + \lambda\Omega(2). $$
$$ \mathbb{E}[\widehat S_{VT}(\lambda)] = \eta_i + (2\lambda - 1)\Omega(2). $$

This confirms the proposition formula.
For any parameter $\lambda \in [0.5, 1)$, the score boost multiplier $(2\lambda - 1)$ is strictly less than 1. Under standard Superscoring ($\lambda=1$), the student expects to harvest the full premium $\Omega(2)$. By setting $\lambda < 1$, the Variance Tax explicitly forces the final score to incorporate the minimum draw. Because the minimum draw carries a negative expected noise shock ($-\Omega(2)$), taking the exam solely to mine favorable noise introduces strict downside risk. This acts as a mathematical variance tax on retaking that scales proportionally with the severity of the expected noise premium.
\end{proof}

\begin{proposition}[Optimal Variance Tax Parameter]
\label{prop:optlambda}
The university's welfare objective $\mathcal{V}(\lambda)=e^*(\lambda)-\theta(2\lambda-1)\bar\Omega(W_i)$, which captures the trade-off between effort incentives and displacement costs, has a unique interior maximum $\lambda^*\in(0.5,1)$. Interiority follows from $\frac{d\mathcal{V}}{d\lambda}\big|_{\lambda=0.5}>0$ and $\frac{d\mathcal{V}}{d\lambda}\big|_{\lambda=1}<0$.
\end{proposition}

\begin{proof}
The Variance Tax forces the university to optimize a continuous trade-off. We model the university as a Principal maximizing a reduced-form objective function $\mathcal{V}(\lambda)$ over the interval $\lambda \in [0.5, 1]$. This function evaluates the tradeoff between the effort response of retaking students (Type B) and the cross-type measurement distortion they generate relative to non-retaking students (Type A):
$$\max_{\lambda \in [0.5,1]} \mathcal{V}(\lambda) = e^*(\lambda) - \theta(2\lambda-1)\bar{\Omega}(\mathcal{W}_i),$$
where $e^*(\lambda)$ is the endogenous optimal effort of a retaking student, $\theta > 0$ captures the university's weight on mitigating cross-type measurement distortion, and $\bar{\Omega}(\mathcal{W}_i)$ is the expected score boost for the wealth-constrained population. 

We first evaluate the student's optimal effort response $e^*(\lambda)$. Assume a student takes the exam in Round 1, revealing baseline ability $\eta_1$. Before retaking in Round 2, the student exerts productive effort $e \ge 0$ at a strictly convex cost $c(e)$ (where $c'(e) > 0, c''(e) > 0$) to increase their true ability to $\eta_2 = \eta_1 + e$. 

For moderate effort where the maximum operator only probabilistically selects Round 2 (the generic interior case), the student's expected final score under the Variance Tax is approximated by:
$$\mathbb{E}[S_{final}(\lambda, e)] \approx \eta_1 + \lambda e + (2\lambda - 1)\Omega.$$
The student chooses effort to maximize expected utility (Score minus Cost):
$$\max_{e} \ U(e) = \eta_1 + \lambda e + (2\lambda - 1)\Omega - c(e).$$

The First-Order Condition (FOC) for the student yields:
$$c'(e^*) = \lambda.$$
By the Implicit Function Theorem, because the marginal cost of effort is strictly increasing ($c''(e) > 0$), the optimal effort response $e^*(\lambda)$ is strictly increasing in $\lambda$: 
$$ \frac{\partial e^*}{\partial \lambda} = \frac{1}{c''(e^*)} > 0. $$

Returning to the university's problem, the Principal's First-Order Condition is:
$$\frac{d\mathcal{V}}{d\lambda} = e^{*'}(\lambda) - 2\theta \bar{\Omega}(\mathcal{W}_i) = 0.$$
Assuming $c(\cdot)$ is sufficiently convex such that $c'''(e) \ge 0$, the effort response function $e^*(\lambda)$ is concave ($e^{*''}(\lambda) \le 0$). This guarantees that the university's objective $\mathcal{V}(\lambda)$ is strictly concave, ensuring that any root of the FOC is a unique global maximum.

We evaluate the marginal benefit of $\lambda$ at the boundaries:
\begin{enumerate}
    \item \textbf{Lower Bound ($\lambda = 0.5$):} The marginal measurement distortion is $2\theta \bar{\Omega}$. Because $e^{*'}(0.5) > 0$, if the university values the initial marginal return on genuine learning more than the baseline noise distortion, $\frac{d\mathcal{V}}{d\lambda} \big|_{\lambda=0.5} > 0$, ensuring $\lambda^* > 0.5$.
    \item \textbf{Upper Bound ($\lambda = 1$):} At standard Superscoring, the marginal cost of the distortion term is strictly positive ($2\theta \bar{\Omega} > 0$). Provided the university places sufficient weight $\theta$ on mitigating demographic inequality and measurement error, the marginal cost of inflation exceeds the marginal benefit of effort: $e^{*'}(1) < 2\theta \bar{\Omega}(\mathcal{W}_i)$. Therefore, $\frac{d\mathcal{V}}{d\lambda}\big|_{\lambda=1} < 0$, ensuring $\lambda^* < 1$.
\end{enumerate}
By the Intermediate Value Theorem, because $\frac{d\mathcal{V}}{d\lambda}$ is continuous, positive at $\lambda=0.5$, and negative at $\lambda=1$, there exists a unique interior optimum $\lambda^* \in (0.5, 1)$ that explicitly balances signal efficiency against effort preservation.

\paragraph{Comparative Statics.}
Applying the Implicit Function Theorem to the university's FOC, $e^{*'}(\lambda^*) - 2\theta \bar{\Omega} = 0$:
\begin{itemize}
    \item \textbf{Decreasing in score noise ($\sigma_\delta^2$):} Higher noise strictly increases the expected boost $\bar{\Omega}$. To restore the FOC equality, $e^{*'}(\lambda^*)$ must increase. Because $e^{*'}(\lambda)$ is a decreasing function (due to $e^{*''} \le 0$), the optimal parameter $\lambda^*$ must decrease. Higher noise forces the university to tax variance more aggressively.
    \item \textbf{Decreasing in cost steepness ($k$):} Let the cost function be parameterized by a steepness scalar $k$, where $c(e, k) = k \cdot \tilde{c}(e)$. A steeper cost curve (higher $k$) makes student effort more inelastic, strictly decreasing the marginal effort response $e^{*'}(\lambda)$. Because the marginal benefit of incentivizing effort has fallen relative to the constant marginal cost of measurement distortion, the university optimizes by decreasing $\lambda^*$. When effort is inelastic, the university abandons effort incentives to prioritize signal fidelity.
\end{itemize}
\end{proof}
\subsection{Contextual Correction}
\label{app:ctx}
FISC and the Variance Tax aim to minimize the mean estimation error. University admissions, however, select strictly from the right tail of the score distribution. As established by \cite{kleinberg2016inherent}, no admissions mechanism can simultaneously equalize False Negative Rates (FNR) and False Positive Rates (FPR) across groups whenever base rates of true qualification differ. Rather than minimizing MSE, the Contextual Correction therefore shifts the institutional objective to target Equal Opportunity: equalizing the FNR across wealth groups for equally qualified applicants.

\begin{proposition}[Contextual Correction]
\label{prop:contextual}
Assume the applicant pool is partitioned into wealth cohorts $q\in\mathcal{G}$. To equalize the False Negative Rate across cohorts, the Equal Opportunity Contextual Correction applies a group-specific penalty $c^*(q)$ to the submitted score: $\widehat\eta_i = S_i - c^*(q)$.

For a target admission threshold $\tau$ and target recall rate $\kappa \in (0,1)$, the exactly feasible penalties satisfy:
\begin{align*}
    P\!\left(\widehat\eta_i\ge\tau \mid Y_i=1,~W_i=q\right)=\kappa \quad\forall\, q\in\mathcal{G}, 
\end{align*}
where $Y_i=1$ indicates true qualification. These penalties are uniquely defined by:
\begin{align*} 
c^*(q) = F_q^{-1}(1-\kappa\mid Y_i=1)-\tau, 
\end{align*}
where $F_q(\cdot\mid Y_i=1)$ is the conditional score CDF within cohort $q$. Existence and uniqueness follow from the strict monotonicity of each $F_q$ and the continuity of the total admitted count in $\tau$.
\end{proposition}
\begin{proof}
Let $F_q(S \mid Y_i = 1)$ denote the Cumulative Distribution Function (CDF) of the submitted Superscore $S_i$ for truly qualified applicants ($Y_i = 1$) within a specific wealth cohort $\mathcal{W}_i = q$.

As established in Proposition~\ref{prop:accelerant}, because the fixed cost of retaking $C$ is strictly decreasing in wealth, higher wealth cohorts optimize by executing a strictly higher attempt frequency ($N_{q'} > N_q$ for $q' > q$). Because the maximum operator over a larger number of independent noise draws First-Order Stochastically Dominates (FOSD) a smaller number of draws, the conditional score distribution of the wealthier cohort is strictly right-shifted. Thus, for any score $S$, the CDF of the wealthier cohort is strictly lower: $F_{q'}(S \mid Y_i=1) < F_q(S \mid Y_i=1)$.

The university seeks a vector of demographic penalties $\mathbf{c} = (c_{q_1}, \dots, c_{q_G})$ applied via $\hat{\eta}_i = S_i - c_q$, and a global admission threshold $\tau$ such that total admissions equal capacity $K$, and the Equal Opportunity constraint (True Positive Rate) is exactly $\kappa$ across all cohorts:
$$ \mathbb{P}(\hat{\eta}_i \ge \tau \mid Y_i = 1, \mathcal{W}_i = q) = \kappa \quad \forall q \in \mathcal{G}. $$

Substituting the estimator $\hat{\eta}_i = S_i - c_q$ yields:
$$ \mathbb{P}(S_i \ge \tau + c_q \mid Y_i = 1, \mathcal{W}_i = q) = \kappa. $$
Because the continuous testing noise generates a continuous score distribution, $\mathbb{P}(S_i \ge x) = 1 - F_q(x)$. Substituting the CDF:
$$ 1 - F_q(\tau + c_q \mid Y_i = 1) = \kappa \implies F_q(\tau + c_q \mid Y_i = 1) = 1 - \kappa. $$

Because the continuous noise distribution has full support, $F_q$ is strictly monotonically increasing, making its inverse $F_q^{-1}$ well-defined and unique. Thus, for any arbitrary target admission threshold $\tau$ and target recall $\kappa$, there exists a unique exact penalty:
$$ c^*(q) = F_q^{-1}(1 - \kappa \mid Y_i = 1) - \tau. $$

Because $F_{q'}(S) < F_q(S)$ by FOSD, evaluating the inverse at the same probability threshold $(1-\kappa)$ guarantees that $F_{q'}^{-1}(1-\kappa) > F_q^{-1}(1-\kappa)$. This correctly forces the optimal penalty $c^*(q)$ to be strictly increasing in the wealth cohort $q$, neutralizing the algorithmic advantage of harvesting variance.

This derivation establishes the unique $c^*(q)$ for any fixed $\tau$. To verify a jointly feasible solution exists that satisfies the university's capacity constraint $K$, observe that for any fixed $\kappa \in (0,1)$, the total expected admitted count is given by:
$$ A(\tau) = \sum_q n_q \cdot \mathbb{P}(\hat\eta_i \ge \tau \mid \mathcal{W}_i = q), $$
where $n_q$ is the population size of cohort $q$.

For each $\tau$, the penalties $c^*(q)$ are instantaneously chosen to satisfy the Equal Opportunity constraint. As $\tau \to \infty$, no students clear the threshold ($A(\tau) \to 0$). As $\tau \to -\infty$, all students clear the threshold ($A(\tau) \to \sum n_q > K$). Because the underlying score distributions are continuous, $A(\tau)$ is a continuous, strictly decreasing function of $\tau$.

By the Intermediate Value Theorem, there exists a unique global threshold $\tau^*$ that achieves exactly capacity $K$. The corresponding vector of penalties $\mathbf{c}^*$, computed at this $\tau^*$, is therefore uniquely determined, establishing the existence and uniqueness of the jointly feasible $(\tau^*, \mathbf{c}^*)$ admission policy.
\end{proof}
\paragraph{Robustness and Implementation Constraints} By conditioning on empirical wealth cohorts $W_i$, this mechanism is robust to violations of noise homogeneity, targeting the realized score distribution within each cohort rather than relying on a population-wide noise model. However, this robustness requires sufficient per-cohort sample sizes to estimate $F_q$ reliably in the tail. Furthermore, applying explicitly different score penalties based on demographic group membership raises disparate-treatment concerns in institutionally or legally constrained environments.

\subsection{Comparing Mechanisms}
Applying Theorem \ref{thm:optimal} to specific information sets explains the structural relationship between these mechanisms. For instance, knowing the exact attempt count ($N_i$) pins down the order statistic exactly, whereas using only a wealth proxy ($W_i$) identifies only the distribution of attempts.

\begin{corollary}[Two Distinct Orderings]
\label{cor:ordering} 
The three mechanisms are ordered along two separate dimensions. 
 
\smallskip 
\noindent \textit{Welfare dominance (admitted ability).} Within the additive correction class of Theorem \ref{thm:optimal}, FISC strictly dominates any mean-unbiased contextual correction, since exact attempt count strictly refines a demographic wealth proxy. The Variance Tax does not belong to this additive class and is therefore not directly comparable on this axis without the parameter analysis of Proposition \ref{prop:optlambda}.
 
\smallskip 
\noindent \textit{Recall equity (Equal Opportunity).} The Contextual Correction is constructed to exactly satisfy $P(\hat Y_i=1\mid Y_i=1,W_i=q)=\kappa$ for every cohort $q$, strictly dominating both the Variance Tax and FISC on this criterion unless those mechanisms are separately tuned to target Equal Opportunity directly.
\end{corollary}
\begin{proof}
We establish the ordering along the two distinct dimensions sequentially.

\textbf{Part 1: Welfare Dominance (Admitted Ability)}
Recall from Section \ref{sec:model} that a student's composite true ability is $\eta_i \equiv w_M\eta_{iM} + w_V\eta_{iV}$. The uncorrected composite Superscore can be expressed as $S_i^{SC} = \eta_i + \Omega(N_i) + \epsilon_i$, where $N_i \in \{1, 2\}$ is the endogenous attempt count, $\Omega(N_i)$ is the expected mechanical noise premium derived from $N_i$ draws, and $\epsilon_i$ is a mean-zero idiosyncratic error. 

By Theorem \ref{thm:optimal}, an additive algorithmic correction takes the form $S_i^{corr} = S_i^{SC} - \mathbb{E}[\Omega(N_i) \mid \mathcal{I}_i]$, where $\mathcal{I}_i$ is the available institutional information set. This minimizes the Mean Squared Error (MSE) between the corrected score and true composite ability.

For the Full-Information Score Correction (FISC), the exact attempt count is observed: $\mathcal{I}_i^{FISC} = \sigma(N_i)$. Because the noise premium is a deterministic function of the attempt count, the penalty is exact and the conditional variance is zero: $\text{Var}(\Omega(N_i) \mid N_i) = 0$.

For the mean-unbiased Contextual Correction, the institution only observes the demographic proxy: $\mathcal{I}_i^{CTX} = \sigma(W_i)$. Because a student's optimal stopping decision $N_i$ depends on stochastic first-round testing shocks ($\delta_1$) in addition to their wealth constraints, $N_i$ remains non-degenerate conditional on $W_i$. Thus, $\text{Var}(\Omega(N_i) \mid W_i) > 0$.

By the Law of Total Variance, the residual estimation error strictly satisfies:
\begin{align*}
    \mathbb{E}\left[ (\Omega(N_i) - \mathbb{E}[\Omega(N_i) \mid W_i])^2 \right] 
    > \mathbb{E}\left[ (\Omega(N_i) - \mathbb{E}[\Omega(N_i) \mid N_i])^2 \right]
\end{align*}

Because FISC strictly minimizes the residual variance of the algorithmic bias, it yields a measurement with a higher signal-to-noise ratio than the demographic proxy. Under uniform threshold admission, a strictly more accurate signal of $\eta_i$ yields an admitted class with strictly higher expected ability. FISC therefore strictly dominates the mean-unbiased contextual correction on the welfare axis. The Variance Tax compresses the order statistics rather than applying an additive penalty; it thus exists outside this bounded class and requires the equilibrium parameter analysis of Proposition \ref{prop:optlambda}.

\textbf{Part 2: Recall Equity (Equal Opportunity)}
Let $Y_i \in \{0,1\}$ denote an indicator for a truly qualified student (e.g., $\eta_i \ge \bar{\eta}$), and let $\hat{Y}_i \in \{0,1\}$ denote admission based on clearing a global threshold $\tau$. Equal Opportunity requires uniform True Positive Rates across all wealth cohorts $q \in \mathcal{G}$:
$$ \mathbb{P}(\hat{Y}_i = 1 \mid Y_i = 1, W_i = q) = \kappa \quad \forall q \in \mathcal{G} $$

By construction, Contextual EqOpp sets cohort-specific penalties $c^*(q)$ (or effectively, cohort-specific thresholds) by evaluating the inverse conditional CDF of the score distribution for each group. It explicitly calibrates the right tail such that $\mathbb{P}(S_i^{SC} - c^*(q) \ge \tau \mid Y_i=1, W_i=q) = \kappa$. It therefore satisfies Equal Opportunity mechanically.

Conversely, FISC and the Variance Tax are demographic-blind. They apply a structurally uniform mapping $\mathcal{M}(S_{i,1}, S_{i,2}, N_i)$ globally. 

Because the fixed retake cost $C(W_i)$ dictates participation capacity, the conditional distribution of attempt counts differs across wealth cohorts. For high-wealth ($q_H$) and low-wealth ($q_L$) applicants of identical true ability:
\begin{align*}
    \mathbb{P}(N_i = 2 \mid Y_i = 1, W_i = q_H)
    > \mathbb{P}(N_i = 2 \mid Y_i = 1, W_i = q_L)
\end{align*}
Because differing $N_i$ distributions map into disparate corrected score spaces, the conditional cumulative distribution functions of the corrected scores strictly diverge:
\begin{align*}
\mathbb{P}(S_i^{FISC} \le s \mid Y_i=1, W_i=q_H) 
\neq \mathbb{P}(S_i^{FISC} \le s \mid Y_i=1, W_i=q_L)
\end{align*}
Evaluating these distinct CDFs at a single, globally uniform admission threshold $\tau$ mathematically guarantees diverging recall rates ($\kappa_{q_H} \neq \kappa_{q_L}$), violating the constraint. Therefore, explicit demographic corrections strictly dominate uniform structural rules on the Equal Opportunity metric.
\end{proof}

\section{Full Experimental Setup}
\label{app:experimental_setup}
This section details the empirical calibration, computational infrastructure, and simulation procedures underlying the experimental results presented in Section~\ref{subsec:results}. Raw SAT section scores $[200, 800]$ and composite scores $[400, 1600]$ are normalized to $[0,1]$. To ensure statistical stability, all reported point estimates and distributions are averaged over $M = 200$ Monte Carlo draws per student for a simulated population of $N=10,000$.

\paragraph{Reproducibility and Infrastructure.} To guarantee exact reproducibility, all stochastic processes---including Monte Carlo noise sampling and population generation---are initialized deterministically (seed $= 42$). The simulation pipeline is implemented in Python 3.10, utilizing NumPy (v1.23) and Numba (v0.56) for heavily optimized, parallelized tensor operations. Experiments were executed on standard consumer-grade multi-core CPU (32GB RAM, Ubuntu Linux), with a complete 10,000-student forward pass requiring approximately 20 minutes. The fully commented source code required to execute the baseline dynamic programming models and replicate all figures and tables is provided in the supplementary material and will be released under an MIT license upon publication.

\subsection{Student Population and Empirical Calibration}
Students $i$ are characterized by a wealth percentile $w_i \sim \text{Uniform}(0,1)$. Their latent Math ability ($\eta_{iM}$) and Verbal ability ($\eta_{iV}$) are drawn jointly from a bivariate normal distribution:
\begin{equation} 
    \begin{pmatrix} \eta_{iM} \\ \eta_{iV} \end{pmatrix}
    \sim
    \mathcal{N}\!\left( 
        \begin{pmatrix} \alpha^M_0 + \alpha^M_1 w_i \\ \alpha^V_0 + \alpha^V_1 w_i \end{pmatrix}, 
        \begin{pmatrix}
        (\sigma_M)^2 & \rho \sigma_M \sigma_V \\
        \rho \sigma_M \sigma_V & (\sigma_V)^2
    \end{pmatrix} 
    \right).
\end{equation}

Using empirical composite and section variances from the \textit{2025 SAT Suite of Assessments Annual Report}, we derive a cross-subject correlation of $\rho \approx 0.81$. The observed real-world gap between the highest and lowest income quintiles is 134 points for Math and 129 points for Verbal. We calibrate the wealth gradients $(\alpha_1^M, \alpha_1^V)$ to these empirical means.

To avoid double-counting the real-world preparation effects already embedded in the empirical data, we apply a structural discount factor. We attribute 40\% of the observed empirical gap to pre-study human capital (structural K-12 inequality), leaving the remaining 60\% to be generated endogenously by the model's heterogeneous effort costs. Score shocks $\varepsilon \sim \mathcal{N}(0, \sigma^2)$ capture test-day variance, with $\sigma \approx 0.2783$ calibrated to account for 60\% of within-group score dispersion.

\subsection{Model Primitives and Procedure}
\paragraph{Effort and Costs}
The deterministic score component for subject $j \in \{M, V\}$ uses a strictly concave knowledge production function $f(\eta_j e_j) = 1 - e^{-\eta_j e_j}$. We strictly enforce convex effort costs ($\alpha = 2$) via the cost function $c(e) = \frac{k_i}{2}(e_M + e_V)^2$. To model dynamic costs, the effort scalar declines linearly in wealth: $k_i = 0.14 - 0.02w_i$.

\paragraph{Retake Friction}
The fixed cost of a second attempt is modeled as a financial constraint inversely proportional to the student's family income. We construct a continuous piecewise linear mapping $\mathcal{I}(w_i)$ from wealth percentiles to US dollars using the exact 2025 SAT income quintile boundaries (bounded by a $\$20,000$ floor and $\$250,000$ ceiling). The fixed retake cost is evaluated as the relative burden of the standard SAT registration fee against this income: $C_i = 100 \times \frac{\$60}{\mathcal{I}(w_i)}$. The baseline retake probability is defined as $p_i = 0.10 + 0.80w_i$, acting as a reduced-form proxy for access to preparation resources and scheduling flexibility. Future payoffs are discounted at $\beta = 0.90$.

\paragraph{Simulation Procedure.}
Students solve the two-stage optimization program via backward induction over a grid of 15 effort levels per subject (spanning $[0.0, 1.5]$ in normalized effort space). The structural scoring rules (Single-Sitting, Superscoring, and the Variance Tax) each induce unique endogenous effort allocations. Conversely, the Full-Information Score Correction (FISC) and Contextual Correction are evaluated as ex-post score transformations applied to the equilibrium effort allocations and retake decisions generated under the Superscoring regime. Final admission outcomes are evaluated primarily via False Negative Rates (FNR) because, in heavily constrained selection pipelines, institutional unfairness manifests overwhelmingly as the exclusionary rejection of highly qualified applicants.

\paragraph{Robustness Grid Sweep}
To empirically validate that our institutional error rates and fairness conclusions are not artifacts of the specific parameter calibration, we execute a comprehensive parameter grid sweep (results detailed in Appendix~\ref{app:grid_robustness}). For this sensitivity analysis, we simulate a population of $N=5,000$ students and evaluate the pipeline across varying temporal discount factors $\beta \in \{0.70, 0.85, 1.0\}$ and test-day noise variance multipliers $\{0.5, 1.0, 2.0\}$, while holding the baseline effort cost multiplier fixed at $1.0$.

\subsection{Calibration of the Baseline Ability Discount}
\label{app:calibration}
The empirical wealth-score gap observed in the College Board data reflects both pre-existing human capital differences---such as school quality and early childhood investment---and the endogenous effects of late-stage test preparation \cite{rothstein2004college, park2015who}. If our simulation assigned 100\% of the observed gap to latent baseline ability, it would implicitly assume that wealth-based effort advantages generate zero score inflation, contradicting the theoretical model.

To accurately isolate the impact of strategic testing behavior, we apply a discount factor to the baseline ability gap. This forces the model to generate the remainder of the observed empirical gap endogenously through our wealth-dependent cost primitives. For reproducibility, we determine the exact magnitude of this discount via a computational calibration sweep.

\paragraph{Calibration Setup.}
We simulate a student population ($N=2000$) where each student draws a wealth percentile $w_i \sim \text{Uniform}(0,1)$ and faces a wealth-dependent marginal cost of effort, parametrized in the normalized score space as $k(w_i) = 0.18 - 0.08w_i$. Latent abilities are drawn from a bivariate normal distribution fitted to the empirical 2025 SAT wealth gradients. To preserve the overall population mean while compressing the baseline variance, we pivot the ability discount around the median wealth ($w=0.50$).

We evaluate the resulting simulated scores against two empirical targets derived from the 2025 College Board report \cite{collegeboard2025}:
\begin{enumerate}
    \item \textbf{The Q5-Q1 SAT Gap:} Targeting the empirical spread of 120--160 points between the highest and lowest wealth quintiles.
    \item \textbf{The Q3 Mean:} Targeting a score of 950--1010 to ensure the overall population proficiency remains anchored to empirical reality.
\end{enumerate}

\paragraph{Heuristic Retake Adjustment} 
Because the calibration sweep evaluates a computationally efficient single-period forward pass, it omits the score inflation generated by strategic retaking and order-statistic selection present in the full model. To accurately align this heuristic proxy with empirical College Board data (which includes final scores after multiple attempts), we apply a flat $+45$ point retake adjustment to the simulated final scores before evaluation against the empirical targets.

\paragraph{Results}
The hyperparameter calibration sweep demonstrates that discounting the baseline ability gap by 35\% to 40\% successfully reproduces the observed 2025 SAT wealth gradient while satisfying both empirical targets. Discounts below 30\% fail to generate a sufficiently wide Q5-Q1 gap under the specified cost primitives. To ensure conservative estimates of our model's endogenous effects, all main simulations reported in Section~\ref{subsec:results} apply a 40\% baseline ability discount.

\subsection{Parameter Grid Robustness}
\label{app:grid_robustness}
To ensure our empirical results are not the consequence of a specific parameter calibration, we evaluate the institutional error rates across a comprehensive parameter grid. We first establish the baseline population distribution, then test the scoring rules against varying structural parameters.

\paragraph{Baseline True Ability Distribution}
The empirical skew of the 2025 SAT applicant pool creates a structurally unbalanced baseline. Table~\ref{tab:pop_distribution} details the latent ability distribution for our simulated population of $N=10,000$ students. Because effort costs are strictly decreasing in wealth, the highest wealth quintile (Q5) heavily saturates the upper tail of latent ability, accounting for 46.4\% of the true top 5\% of applicants. This establishes the absolute ground truth against which all False Negative Rates (FNR) are measured.

\begin{table}[h]
\centering
\begin{tabular}{lccc}
\toprule
\textbf{Q} & \textbf{\# Students} & \textbf{Top 5\% (Count)} & \textbf{Top 5\% (Share)} \\
\midrule
Q1 & 1,256 & 25 & 5.0\% \\
Q2 & 1,409 & 53 & 10.6\% \\
Q3 & 1,670 & 81 & 16.2\% \\
Q4 & 2,284 & 109 & 21.8\% \\
Q5 & 3,381 & 232 & 46.4\% \\
\bottomrule
\end{tabular}
\caption{\footnotesize Global population and true latent ability distribution. The top 5\% share represents the baseline for equitable representation prior to the application of any scoring mechanism.}
\label{tab:pop_distribution}
\end{table}

\paragraph{Parameter Grid Evaluation}
We vary the core structural parameters that govern student behavior and score generation: the temporal discount factor $\beta \in \{0.70, 0.85, 1.0\}$ and the test-day noise variance $\sigma^2$ (scaled by multipliers of $0.5$, $1.0$, and $2.0$). The baseline effort cost multiplier $k$ is held fixed to isolate the interactive effects of testing noise and temporal preferences.

For each parameter combination in the grid, we simulate the full two-stage student optimization problem and record the resulting False Negative Rates. Table~\ref{tab:grid_robustness} reports the mean error rates and standard deviations for all five wealth quintiles, averaged across all grid configurations. 

\begin{table*}[t]
\centering
\begin{tabular}{lccccc}
\toprule
\textbf{Rule} & \textbf{Q1 FNR} & \textbf{Q2 FNR} & \textbf{Q3 FNR} & \textbf{Q4 FNR} & \textbf{Q5 FNR} \\
\midrule
Single-Sitting (SS) & 76.4\% $\pm$ 9.7 & 35.6\% $\pm$ 8.2 & 21.1\% $\pm$ 7.6 & 10.3\% $\pm$ 4.0 & 7.9\% $\pm$ 5.4 \\
Superscoring (SC)   & 68.4\% $\pm$ 12.7 & 31.5\% $\pm$ 7.4 & 23.7\% $\pm$ 8.1 & 12.7\% $\pm$ 5.7 & 9.4\% $\pm$ 5.3 \\
Variance Tax (TAX)  & 66.7\% $\pm$ 6.3 & 33.3\% $\pm$ 8.1 & 18.9\% $\pm$ 7.8 & 8.4\% $\pm$ 4.1 & 5.3\% $\pm$ 4.1 \\
FISC                & 60.9\% $\pm$ 14.8 & 25.2\% $\pm$ 6.0 & 20.1\% $\pm$ 6.7 & 11.5\% $\pm$ 4.9 & 9.6\% $\pm$ 5.5 \\
Contextual (CTX)    & 25.3\% $\pm$ 11.5 & 15.7\% $\pm$ 7.7 & 14.7\% $\pm$ 6.0 & 18.4\% $\pm$ 7.6 & 16.3\% $\pm$ 8.5 \\
\bottomrule
\end{tabular}
\caption{\footnotesize Robustness of False Negative Rates across the parameter grid. Contextual Correction maintains an equitable, narrowly distributed FNR band despite high variance in underlying structural parameters, while Single-Sitting severely and disproportionately penalizes the lowest wealth quintiles.}
\label{tab:grid_robustness}
\end{table*}

The aggregated results demonstrate that the ordinal ranking of the policies' equity outcomes is invariant to the structural parameters. Across all extreme realities of noise and discounting, Single-Sitting (SS) uniquely excludes the highest proportion of qualified Q1 students. Demographic-blind measurement corrections (FISC and TAX) offer only marginal improvements. The Contextual Correction (CTX) is the only intervention that consistently neutralizes the algorithmic bias, ensuring stable and equitable error rates across all income strata.

\section{Uniform Population Baselines}
\label{app:uniformpop}
In the main text, our simulations use an empirically calibrated population that reflects the true demographic skew of the 2025 College Board data. To separate the structural mechanics of the scoring rules from this empirical skew, we present two baseline experiments using a uniform population distribution, where each of the five wealth quintiles represents exactly 20\% of the applicant pool. We evaluate this balanced population under two distinct effort-cost regimes to isolate the precise impact of wealth-dependent friction.

\subsection{Baseline 1: Uniform Population and Static Variable Costs}
\label{app:uniform_static}
In this regime, we isolate the structural impact of the scoring rules from the empirical ability skew of the 2025 applicant pool. We simulate a uniform population where each wealth quintile represents roughly 20\% of the true top ability distribution (Table~\ref{tab:app_static_admissions}). To isolate the pure effect of retake affordability, we remove disparities in learning capacity by setting the marginal cost of effort equal across all groups ($k_i = k$), but maintain the empirical wealth-dependent fixed cost of retaking ($C(W_i)$).

\begin{table}[h]
\centering
\begin{tabular}{lcccccc}
\toprule
\textbf{Q} & \textbf{TA} & \textbf{SS} & \textbf{SC} & \textbf{TAX} & \textbf{FISC} & \textbf{CTX} \\
\midrule
Q1 & 23.2\% & 12.8\% & 13.4\% & 15.4\% & 17.8\% & 21.8\% \\
Q2 & 19.6\% & 20.0\% & 21.0\% & 18.8\% & 20.6\% & 20.8\% \\
Q3 & 20.8\% & 24.2\% & 23.2\% & 23.0\% & 23.6\% & 24.0\% \\
Q4 & 19.2\% & 21.4\% & 21.6\% & 21.2\% & 19.8\% & 18.2\% \\
Q5 & 17.2\% & 21.6\% & 20.8\% & 21.6\% & 18.2\% & 15.2\% \\
\bottomrule
\end{tabular}
\caption{\footnotesize [Uniform Baseline] Admitted Class Composition compared to True Ability (TA). Only Contextual Correction aligns admissions with underlying merit.}
\label{tab:app_static_admissions}
\end{table}

\paragraph{Effort and Participation.} 
Even when all students share the identical capacity to convert effort into knowledge, the fixed financial barrier to securing a second attempt severely fractures student behavior. As shown in Figure~\ref{fig:app_uniform_retake}, financial constraints prevent low-wealth students from accessing the algorithmic safety net: Q5 students retake at a 95.5\% rate under Superscoring, while Q1 students retake at only 56.3\%. Consequently, while Round 1 effort is relatively uniform, Q1 is forced to systematically under-invest in Round 2 preparation compared to their unconstrained Q5 peers (Figure~\ref{fig:app_static_effort}).
\begin{figure}[ht] 
\centering 
\includegraphics[width=0.7\columnwidth]{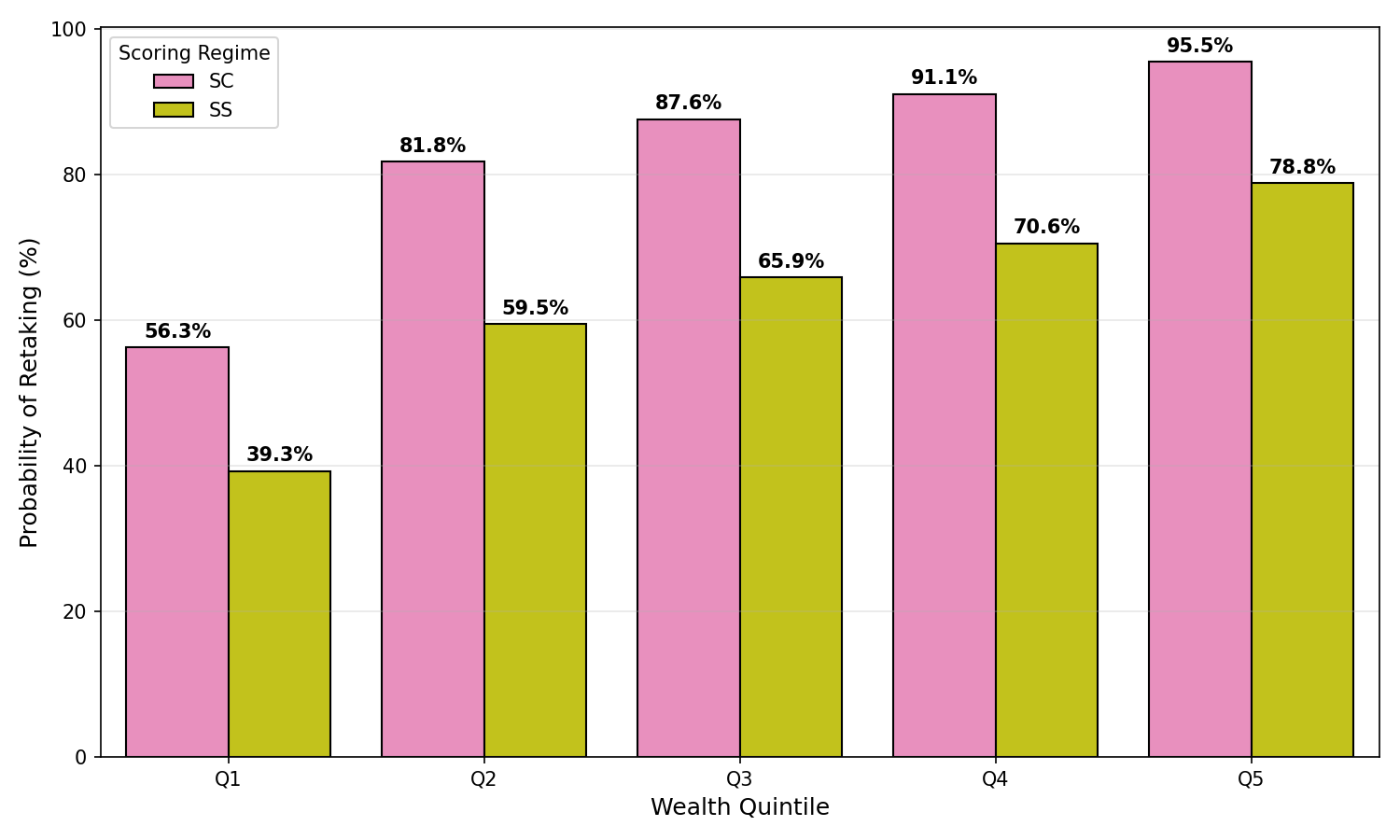}  
\caption{\footnotesize [Uniform Baseline] Probability of retaking under Superscoring (pink) vs. Single-Sitting (yellow) across wealth quintiles.}
\label{fig:app_uniform_retake}
\end{figure}

\begin{figure}[ht] 
\centering 
\includegraphics[width=0.7\columnwidth]{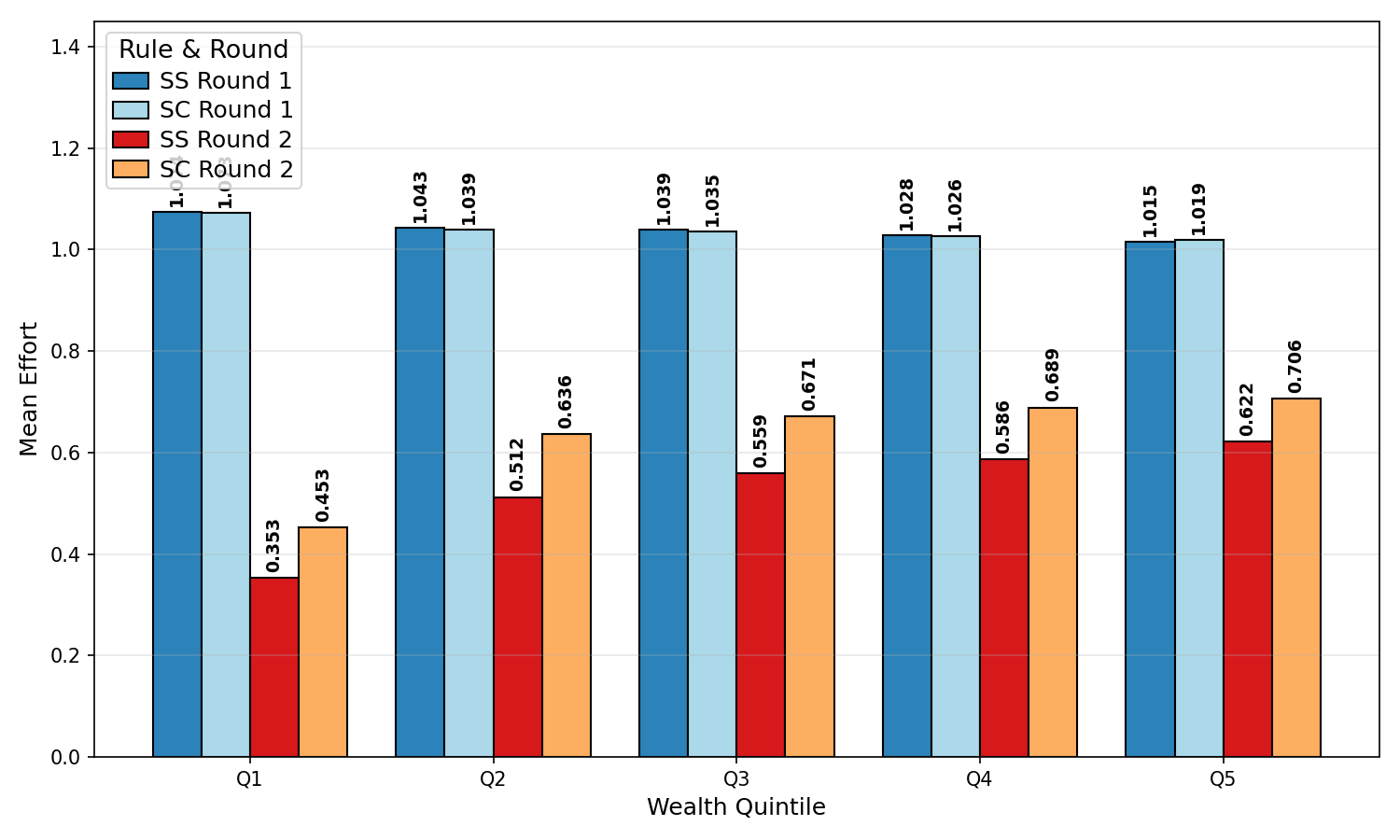}  
\caption{\footnotesize [Uniform Baseline] Mean effort exerted by round and wealth quintile under Single-Sitting in Round 1 (blue, with SS on the left and SC on the right) and Round 2 (red, with SS on the left and SC on the right). Despite uniform marginal costs, fixed retake barriers suppress Q1's Round 2 effort.
Effort Exerted by Round and Wealth Quintile}
\label{fig:app_static_effort}
\end{figure}

\paragraph{Measurement and Institutional Fairness.} 
Because Q1 applicants are financially locked out of multiple draws, they cannot harvest the mechanical noise premium of Superscoring. Figure~\ref{fig:app_uniform_premium} demonstrates that above-median wealth students extract a strictly higher algorithmic score boost than equally qualified below-median wealth peers. 

\begin{figure}[ht] 
\centering 
\includegraphics[width=0.7\columnwidth]{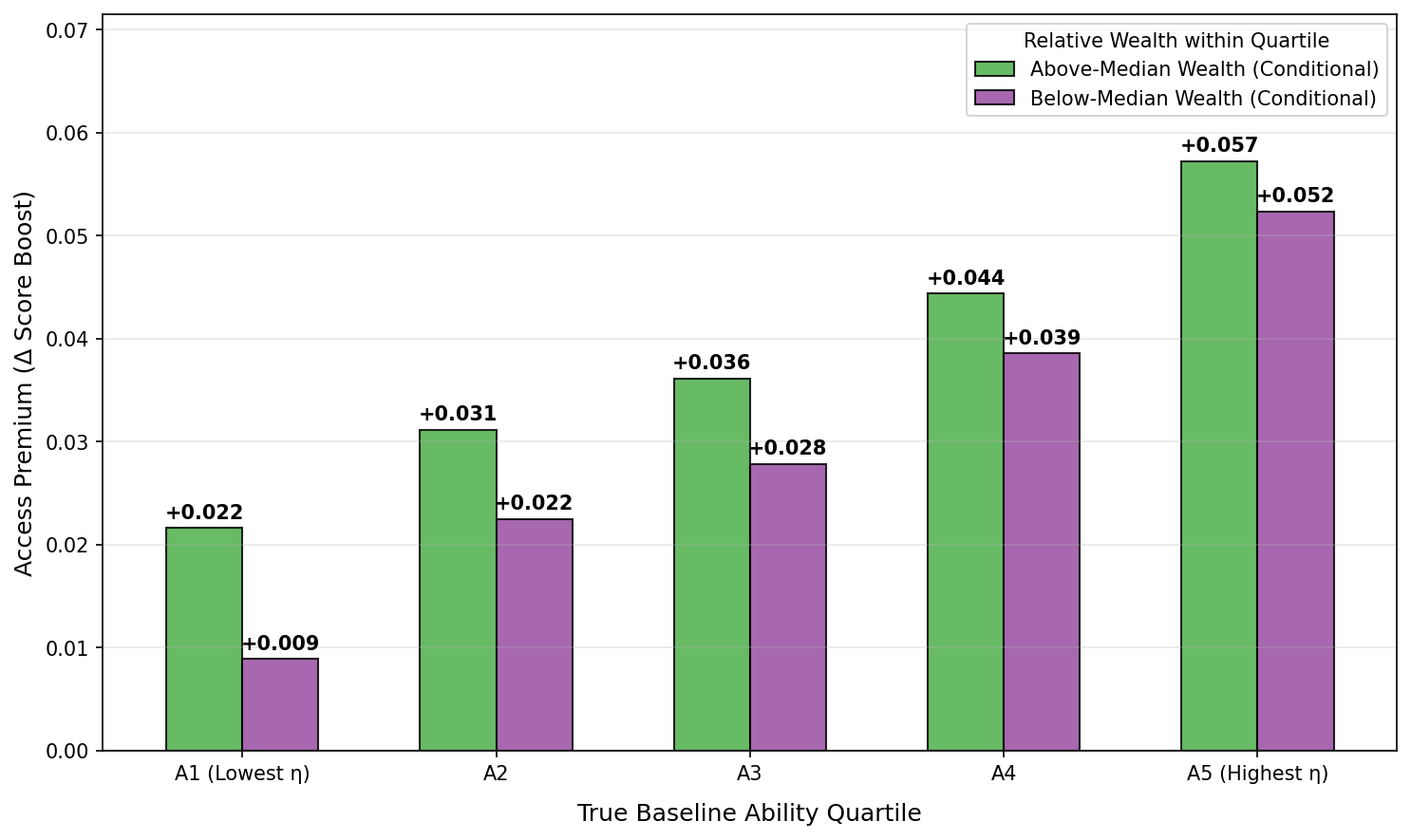}  
\caption{\footnotesize [Uniform Baseline] Conditionally Balanced Access Premium by True Ability Quartile. Wealthier applicants (green) extract higher score boosts at every underlying ability level than below median wealth applicants (purple).}
\label{fig:app_uniform_premium}
\end{figure}

When evaluated for admission, this unearned Access Premium violently distorts the threshold. Despite comprising 23.2\% of the true top ability pool in this balanced simulation, Q1 faces a staggering 48.3\% False Negative Rate under Superscoring (Table~\ref{tab:app_static_fnr}), capturing only 13.4\% of the admitted class (Table~\ref{tab:app_static_admissions}). Contextual Correction (CTX) is the only mechanism that evaluates students against their local structural constraints, dropping the Q1 FNR to 23.3\% and successfully restoring their admitted share to an ability-fair 21.8\%.

\begin{table}[h]
\centering
\begin{tabular}{lccccc}
\toprule
\textbf{Rule} & \textbf{Q1} & \textbf{Q2} & \textbf{Q3} & \textbf{Q4} & \textbf{Q5} \\
\midrule
SS         & 49.1\% & 14.3\% & 7.7\% & 9.4\% & 2.3\% \\
SC         & 48.3\% & 11.2\% & 10.6\% & 8.3\% & 7.0\% \\
TAX        & 36.2\% & 12.2\% & 8.7\% & 6.2\% & 1.2\% \\
FISC       & 29.3\% & 7.1\% & 7.7\% & 9.4\% & 10.5\% \\
CTX        & 23.3\% & 12.2\% & 9.6\% & 14.6\% & 19.8\% \\
\bottomrule
\end{tabular}
\caption{\footnotesize [Uniform Baseline] False Negative Rate (FNR). Even with uniformly distributed ability, fixed financial barriers cause severe exclusion rates for Q1.}\label{tab:app_static_fnr}
\end{table}

Furthermore, as observed in the main text, the algorithmic distortions of Superscoring manifest primarily as exclusionary False Negatives rather than unearned inclusion; False Positive Rates (FPR) remain structurally low (under 1.5\%) across all policies and demographic groups (Table~\ref{tab:app_static_fpr}).

\begin{table}[h]
\centering
\begin{tabular}{lccccc}
\toprule
\textbf{Rule} & \textbf{Q1} & \textbf{Q2} & \textbf{Q3} & \textbf{Q4} & \textbf{Q5} \\
\midrule
SS         & 0.3\% & 0.8\% & 1.3\% & 1.1\% & 1.3\% \\
SC         & 0.4\% & 0.9\% & 1.2\% & 1.1\% & 1.3\% \\
TAX        & 0.2\% & 0.4\% & 1.0\% & 0.9\% & 1.2\% \\
FISC       & 0.4\% & 0.6\% & 1.1\% & 0.7\% & 0.7\% \\
CTX        & 1.0\% & 0.9\% & 1.3\% & 0.5\% & 0.4\% \\
\bottomrule
\end{tabular}
\caption{\footnotesize [Uniform Baseline] False Positive Rate (FPR). Unearned inclusion rates remain structurally low across all scoring mechanisms.}\label{tab:app_static_fpr}
\end{table}

\subsection{Baseline 2: Dynamic Costs (Wealth-Dependent)}
In this regime, we maintain the uniform population distribution but reintroduce the fully wealth-dependent marginal and fixed cost parameters ($k(W_i)$ and $C(W_i)$) used in the main text. This isolates how financial frictions alone drive disparate outcomes, completely independent of the empirical baseline ability skew.



\begin{figure}[ht]
\centering

\begin{minipage}[t]{0.48\columnwidth}
    \centering
    \includegraphics[width=\linewidth]{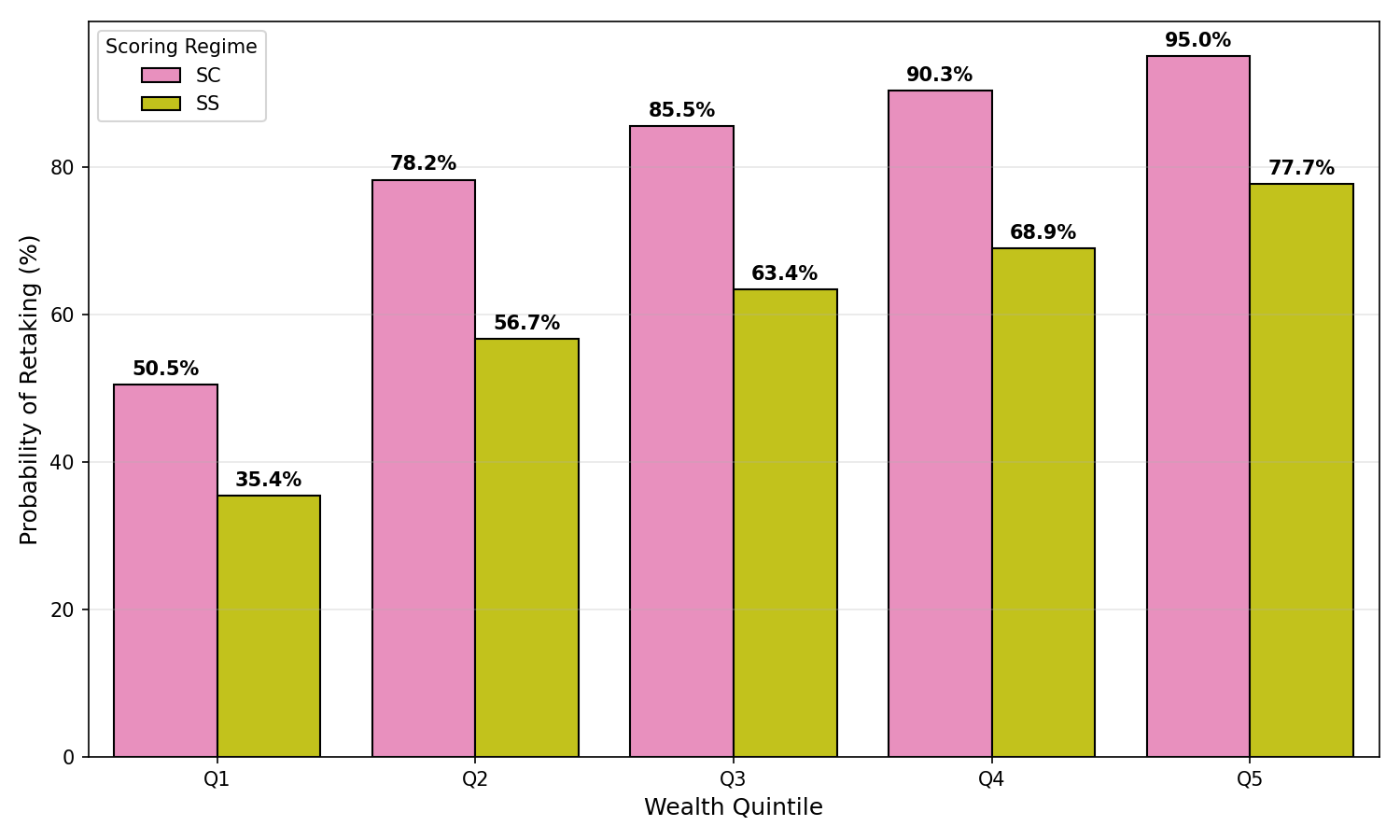}
    \caption{\footnotesize [Dynamic Costs] Probability of retaking under Superscoring (pink) vs. Single-Sitting (yellow) across wealth quintiles. Wealth-dependent costs immediately fracture participation.}
    \label{fig:app_dynamic_retake}
\end{minipage}
\hfill
\begin{minipage}[t]{0.48\columnwidth}
    \centering
    \includegraphics[width=\linewidth]{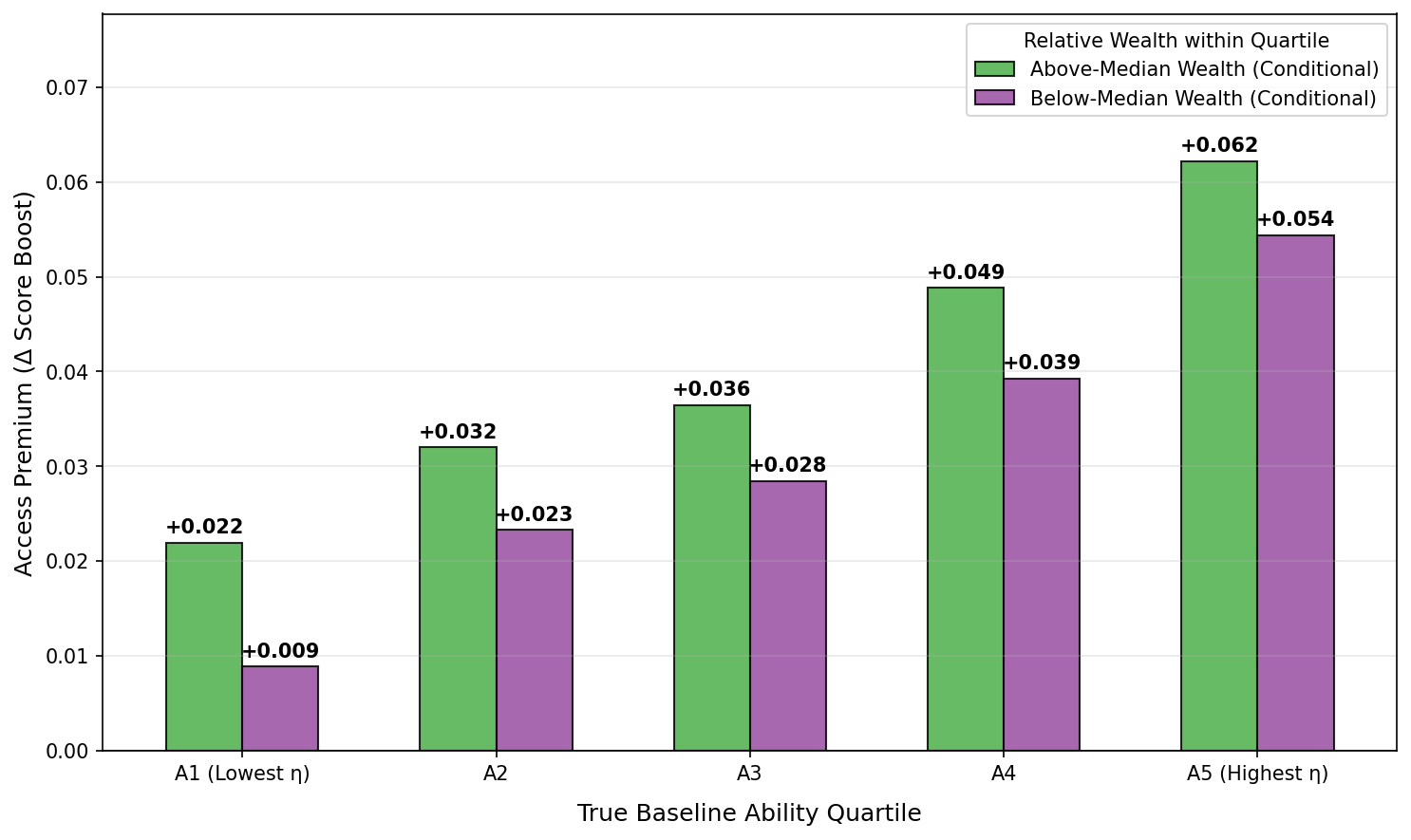}
    \caption{\footnotesize [Dynamic Costs] Conditionally Balanced Access Premium by True Baseline Ability Quartile. Wealthier applicants extract higher algorithmic score boosts.}
    \label{fig:app_dynamic_premium}
\end{minipage}

\end{figure}

\paragraph{The Return of Structural Disparity.} As soon as preparation and retake costs become wealth-dependent, the disparities observed in the 2025 empirical data immediately re-emerge, despite the population being perfectly balanced in latent ability. As shown in Figure~\ref{fig:app_dynamic_retake}, Q5 students utilize the Superscoring retake option at a rate of 95.0\%, while Q1 participation drops to just 50.5\%. Because this unequal access to noise-harvesting artificially inflates Q5 scores, it resurrects the Access Premium across every quartile of underlying ability (Figure~\ref{fig:app_dynamic_premium}).

\begin{table}[ht]
\centering

\begin{minipage}[t]{0.48\columnwidth}
\centering
\begin{tabular}{lccccc}
\toprule
\textbf{Rule} & \textbf{Q1} & \textbf{Q2} & \textbf{Q3} & \textbf{Q4} & \textbf{Q5} \\
\midrule
SS   & 59.5\% & 16.3\% & 6.7\%  & 6.2\%  & 2.3\% \\
SC   & 57.8\% & 14.3\% & 9.6\%  & 5.2\%  & 1.2\% \\
TAX  & 50.9\% & 9.2\%  & 2.9\%  & 2.1\%  & 0.0\% \\
FISC & 44.0\% & 16.3\% & 5.8\%  & 3.1\%  & 3.5\% \\
CTX  & 17.2\% & 12.2\% & 13.5\% & 14.6\% & 12.8\% \\
\bottomrule
\end{tabular}
\captionof{table}{\footnotesize [Dynamic Costs] False Negative Rate (FNR). Economic frictions cause severe exclusion rates for Q1, regardless of the uniform ability distribution.}
\label{tab:app_dynamic_fnr}
\end{minipage}
\hfill
\begin{minipage}[t]{0.48\columnwidth}
\centering
\begin{tabular}{lccccc}
\toprule
\textbf{Rule} & \textbf{Q1} & \textbf{Q2} & \textbf{Q3} & \textbf{Q4} & \textbf{Q5} \\
\midrule
SS   & 0.2\% & 0.3\% & 1.1\% & 1.6\% & 2.1\% \\
SC   & 0.1\% & 0.5\% & 1.0\% & 1.5\% & 2.1\% \\
TAX  & 0.1\% & 0.2\% & 0.7\% & 1.0\% & 2.0\% \\
FISC & 0.1\% & 0.3\% & 0.9\% & 1.0\% & 1.8\% \\
CTX  & 1.5\% & 1.0\% & 0.7\% & 0.3\% & 0.3\% \\
\bottomrule
\end{tabular}
\captionof{table}{\footnotesize [Dynamic Costs] False Positive Rate (FPR). Unearned inclusion rates remain structurally low across all scoring mechanisms and demographic groups.}
\label{tab:app_dynamic_fpr}
\end{minipage}

\end{table}



\begin{table}[h]
\centering
\begin{tabular}{lcccccc}
\toprule
\textbf{Q} & \textbf{TA} & \textbf{SS} & \textbf{SC} & \textbf{TAX} & \textbf{FISC} & \textbf{CTX} \\
\midrule
Q1 & 23.2\% & 10.2\% & 10.0\% & 11.6\% & 13.4\% & 24.8\% \\
Q2 & 19.6\% & 17.4\% & 18.8\% & 18.4\% & 17.6\% & 21.0\% \\
Q3 & 20.8\% & 23.8\% & 22.6\% & 23.0\% & 23.2\% & 20.8\% \\
Q4 & 19.2\% & 24.0\% & 23.6\% & 22.4\% & 22.4\% & 17.4\% \\
Q5 & 17.2\% & 24.6\% & 25.0\% & 24.6\% & 23.4\% & 16.0\% \\
\bottomrule
\end{tabular}
\caption{\footnotesize [Dynamic Costs] Admitted Class Composition. Q5 heavily over-indexes its 17.2\% True Ability (TA) share under SC, corrected only by CTX.}
\label{tab:app_dynamic_admissions}
\end{table}

\paragraph{Measurement and Institutional Fairness.} The structural advantage of affordable retakes severely distorts the admission thresholds. Under uncorrected Superscoring, the FNR for Q1 spikes to 57.8\%, while remaining at a mere 1.2\% for Q5 (Table~\ref{tab:app_dynamic_fnr}). Consequently, despite Q1 comprising 23.2\% of the top latent ability pool in this specific simulation draw, they capture only 10.0\% of the admitted seats under Superscoring, while Q5 captures 25.0\% (Table~\ref{tab:app_dynamic_admissions}). As in the main text, only Contextual Correction (CTX) successfully redistributes the admitted seats back to closely match the true ability baseline (24.8\% for Q1 and 16.0\% for Q5). Furthermore, consistent with all prior results, the False Positive Rates (FPR) remain structurally low (under 2.5\%) across all policies and demographics (Table~\ref{tab:app_dynamic_fpr}).

\section{Empirical Population and Dynamic Costs Graphs and Tables}
\subsection{False Positive Rates by Rule and Wealth Quintile}
\label{subsec:fpr}
\begin{table}[H]
\centering
\begin{tabular}{lccccc}
\toprule
\textbf{Rule} & \textbf{Q1} & \textbf{Q2} & \textbf{Q3} & \textbf{Q4} & \textbf{Q5} \\
\midrule
SS         & 0.0\%& 0.0\%& 0.1\%& 0.9\%& 1.6\%\\
SC         & 0.1\%& 0.1\%& 0.1\%& 1.0\%& 1.6\%\\
TAX        & 0.0\%& 0.0\%& 0.2\%& 0.5\%& 1.4\%\\
FISC       & 0.2\%& 0.0\%& 0.3\%& 0.9\%& 1.4\%\\
CTX& 0.7\&& 0.4\%& 0.2\%& 0.7\&& 1.1\%\\
\bottomrule
\end{tabular}
\caption{\footnotesize False Postive Rate (FPR) by scoring rule for each wealth quintile. The evaluated policies are abbreviated as Single-Sitting (SS), Superscoring (SC), Variance Tax (TAX), Full-Information Score Correction (FISC), and Contextual Correction (CTX).}
\label{tab:error_rates_fpr}
\end{table}

\end{document}